\documentclass[a4paper,UKenglish,cleveref, autoref, thm-restate, numberwithinsect,final]{lipics-v2021}
\hideLIPIcs

\nolinenumbers

\usepackage[T1]{fontenc}

\usepackage[author=anonymous,marginclue,footnote]{fixme}

\FXRegisterAuthor{jf}{ajf}{JF}
\FXRegisterAuthor{pw}{apw}{PW}
\FXRegisterAuthor{ls}{als}{LS}
\FXRegisterAuthor{hb}{ahb}{HB}

\usepackage{tikz-cd}
\usepackage{graphicx}
\usepackage{amsfonts}
\usepackage{amsmath}
\usepackage{amssymb}
\usepackage{mathtools}
\usepackage{stmaryrd}
\usepackage{enumitem}
\usepackage{float}
\usepackage{thmtools}

\usepackage[strings]{underscore}
\usepackage{hyperref}
\hypersetup{pdftex,colorlinks=true,allcolors=blue}

\newcommand{\by}[1]{\qquad\hfill(\text{#1})}
\newcommand{\stream}{\mathsf{stream}}
\newcommand{\ptrace}{\mathsf{ptrace}}

\newcommand{\unif}{\mathsf{unif}}
\newcommand{\Gprob}{{G^\mathsf{prob}}}

\newcommand{\Trace}{\mathit{Tr}}
\newcommand{\id}{\mathit{id}}
\newcommand{\Id}{\mathit{Id}}
\newcommand{\ev}{\mathit{ev}}
\newcommand{\Met}{\mathbf{Met}}
\newcommand{\FPow}{\mathcal{F}}
\newcommand{\FHaus}{\overline{\mathcal{F}}}

\newcommand{\V}{[0,1]}
\newcommand{\qeq}{\mathcal{E}}
\newcommand{\catC}{\mathcal{C}}

\newcommand{\ftFinK}{\overline{\mathcal{D}}_\omega}
\newcommand{\ftDist}{\mathcal{D}_\omega}

\newcommand{\nat}{\mathbb{N}}
\newcommand{\SET}{\mathbf{Set}}
\newcommand{\MET}{\mathbf{Met}}

\newcommand{\Pmet}{\overline{\mathcal{P}}_\omega}

\newcommand{\Pset}{\mathcal{P}_\omega}
\newcommand{\Pc}{\mathcal{P}_{c}}

\newcommand{\sem}[1]{\llbracket #1 \rrbracket}
\newcommand{\rsem}[1]{\llparenthesis #1 \rrparenthesis}
\newcommand{\modal}[1]{\Diamond_{#1}}
\newcommand{\appModal}[2]{\rsem{#1}^{\bullet}(#2)}
\newcommand{\Galg}[2]{\text{Alg}_#1(#2)}

\newcommand{\lrule}[3]{(\mathbf{#1})\;\infrule{#2}{#3}}
\newcommand{\infrule}[2]{\frac{#1}{#2}}

\newcommand{\monad}{\mathbb{M}}

\theoremstyle{definition}
\newtheorem{defn}[theorem]{Definition}

\newcommand{\autorefexpls}[1]{\renewcommand{\exampleautorefname}{Examples}%
\autoref{#1}\renewcommand{\exampleautorefname}{Example}%
}
\newcommand{\autorefsecs}[1]{\renewcommand{\sectionautorefname}{Sections}%
\autoref{#1}\renewcommand{\sectionautorefname}{Section}%
}

\newcommand{\Act}{\mathcal{A}}

\numberwithin{equation}{section}

\title{Quantitative Graded Semantics and\\ Spectra of Behavioural Metrics}

\relatedversion{}
\relatedversiondetails{Conference Version}{https://doi.org/10.4230/LIPIcs.CSL.2025.33}

\titlerunning{Quantitative Graded Semantics and Spectra of Behavioural Metrics}

\ArticleNo{33}

\author{Jonas Forster}{Friedrich-Alexander-Universität Erlangen-Nürnberg,
  Germany}{jonas.forster@fau.de}{https://orcid.org/0000-0002-5050-2565}{}

\author{Lutz Schr{\"o}der}{Friedrich-Alexander-Universität Erlangen-Nürnberg,
  Germany}{lutz.schroeder@fau.de}{https://orcid.org/0000-0002-3146-5906}{}

\author{Paul Wild}{Friedrich-Alexander-Universität Erlangen-Nürnberg,
  Germany}{paul.wild@fau.de}{https://orcid.org/0000-0001-9796-9675}{}

\author{Harsh Beohar}{University of Sheffield, United
Kingdom}{h.beohar@sheffield.ac.uk}{https://orcid.org/0000-0001-5256-1334}{}

\author{Sebastian Gurke}{Universität Duisburg-Essen,
Germany}{sebastian.gurke@uni-due.de}{https://orcid.org/0009-0008-4343-1384}{}

\author{Barbara K{\"o}nig}{Universität Duisburg-Essen,
Germany}{barbara\_koenig@uni-due.de}{https://orcid.org/0000-0002-4193-2889}{}

\author{Karla Messing}{Universität Duisburg-Essen,
Germany}{karla.messing@uni-due.de}{https://orcid.org/0009-0003-1019-6449}{}

\funding{The fourth author was supported by the EPSRC NIA Grant
  EP/X019373/1, while the remaining authors were supported by the
  Deutsche Forschungsgemeinschaft (DFG, German Research Foundation) --
  project number 434050016 (SpeQt).}

\authorrunning{J. Forster, \and P. Wild, \and L. Schr\"oder, \and
H. Beohar, \and S. Gurke, \and B. K{\"o}nig, \and and K. Messing}

\Copyright{Jonas Forster, Paul Wild, Lutz Schr\"oder,
Harsh Beohar, Sebastian Gurke, Barbara K{\"o}nig, and Karla Messing}


\begin{document}

\renewcommand{\subsectionautorefname}{Section}
\renewcommand{\subsubsectionautorefname}{Section}

\maketitle              %

\begin{abstract}
  Behavioural metrics provide a quantitative refinement of classical
  two-valued behavioural equivalences on systems with quantitative
  data, such as metric or probabilistic transition systems. In analogy
  to the linear-time/\,branching-time spectrum of two-valued
  behavioural equivalences on transition systems, behavioural metrics
  vary in granularity, and are often characterized by fragments of
  suitable modal logics. In the latter respect, the quantitative case
  is, however, more involved than the two-valued one; in fact, we show
  that probabilistic metric trace distance cannot be characterized by
  any compositionally defined modal logic with unary modalities. We go
  on to provide a unifying treatment of spectra of behavioural metrics
  in the emerging framework of graded monads, working in coalgebraic
  generality, that is, parametrically in the system type. In the
  ensuing development of \emph{quantitative graded semantics}, we
  introduce algebraic presentations of graded monads on the category
  of metric spaces. Moreover, we provide a general criterion for a
  given real-valued modal logic to characterize a given behavioural
  distance. As a case study, we apply this criterion to obtain a new
  characteristic modal logic for trace distance in fuzzy metric
  transition systems. %
\end{abstract}
\section{Introduction}\label{sec:intro}
\noindent While qualitative models of concurrent systems are
traditionally analysed using various notions of two-valued process
equivalence, it has long been recognized that for systems involving
quantitative data, notions of behavioural \emph{distance} play a
useful role as a more fine-grained measure of process
similarity. Well-known examples include behavioural distances on
probabilistic transition
systems~\cite{GiacaloneEA90,dgjp:metrics-labelled-markov,bw:behavioural-pseudometric},
on systems combining nondeterminism and
probability~\cite{cgt:logical-bisim-metrics}, and on metric transition
systems~\cite{afs:linear-branching-metrics,FahrenbergEA11}.  Like in
the two-valued case, where process equivalences of varying granularity
are arranged on the \emph{linear-time/branching-time
  spectrum}~\cite{DBLP:books/el/01/Glabbeek01}, one has a spectrum of
behavioural metrics on a given system type that vary in granularity
(with greater distances thought of as having finer
granularity)~\cite{DBLP:journals/tcs/FahrenbergL14}.

An important point of interest in this context are
\emph{characteristic modal logics}\lsnote{Apparently the terminology is controversial; other ideas?}. In the two-valued setting, a logic
is \emph{characteristic} for a given behavioural equivalence if the
latter coincides with the respective induced logical
indistinguishability relation, so that behavioural inequivalence can
be certified by distinguishing formulae (as in the recent proof of the
failure of unlinkability in the {ICAO} 9303 e-passport
standard~\cite{DBLP:conf/esorics/FilimonovHMS19})). For instance,
Hennessy-Milner logic is characteristic for
bisimilarity~\cite{DBLP:journals/jacm/HennessyM85}, and most
equivalences on the classical spectrum are characterized by fragments
of Hennessy-Milner logic~\cite{DBLP:books/el/01/Glabbeek01} that are
compositionally defined, i.e.\ given by a choice of modalities and
propositional operators equipped with a recursively defined semantics
(e.g.\ trace equivalence is characterized by the logic built from
diamonds, truth, and -- optionally -- disjunction). In the
quantitative setting, a logic is \emph{characteristic} if the induced
logical distance coincides with the respective behavioural distance,
so that high behavioural distance may be \emph{certified} by means of
distinguishing modal
formulae~\cite{RadyBreugel23}. %
A prototypical example is quantitative probabilistic modal logic,
which is characteristic for branching-time behavioural distance on
probabilistic transition systems~\cite{bw:behavioural-pseudometric}.
However, it turns out that in general, the quantitative setting
behaves less smoothly in this respect than the two-valued
setting. Indeed, we show as our first main result that for
probabilistic metric trace distance (on generative probabilistic
transition systems in which the set of labels is equipped with a
metric, i.e.\ on the probabilistic variant of metric transition
systems), there does not exist any characteristic quantitative modal
logic at all. Here, the term \emph{modal logic} is understood in a
fairly broad sense; essentially, we stipulate no more than that, in
analogy to the two-valued case as discussed above, the logic should be
a compositionally defined fragment of a bisimulation-invariant
next-step logic with unary modalities.

We subsequently work towards positive results, using the framework of
\emph{graded
  semantics}~\cite{DBLP:conf/calco/MiliusPS15,DBLP:conf/concur/DorschMS19}
to achieve an appropriate level of generality. Graded semantics is
parametric both in the \emph{type} of systems (e.g.\ probabilistic,
metric, fuzzy) and in the quantitative \emph{semantics} of systems,
i.e.\ the choice of behavioural distance. The system type is
abstracted as an endofunctor on a suitable base category following the
paradigm of \emph{universal
  coalgebra}~\cite{DBLP:journals/tcs/Rutten00}. Parametricity in the
system semantics, on the other hand, is based on the choice of a
\emph{graded monad}, which handles additional semantic identifications
(beyond branching-time equivalence) by algebraic means, using grades
to control the depth of look-ahead. Both Kleisli-style coalgebraic
trace semantics~\cite{DBLP:journals/lmcs/HasuoJS07} and the smoother,
but less widely applicable Eilenberg-Moore-style coalgebraic trace
semantics~\cite{DBLP:journals/jcss/Jacobs0S15} are subsumed by this
framework~\cite{DBLP:conf/calco/MiliusPS15}. %

Graded semantics has recently been extended to cover behavioural
distances in the Eilenberg-Moore-style
setting~\cite{DBLP:conf/stacs/BeoharG0MFSW24,DBLP:journals/corr/abs-2307-14826},
and, generalizing the two-valued
case~\cite{DBLP:conf/calco/MiliusPS15,DBLP:conf/concur/DorschMS19}, a
canonical notion of \emph{quantitative graded logic} has been
identified. Quantitative graded logics are always \emph{invariant}
under the underlying behavioural distance in the sense that formula
evaluation is nonexpansive, so that logical distance is below
behavioural distance. In some cases, the reverse inequality, i.e.\
\emph{expressivity} of quantitative graded logics, can be established
by a straightforward generalization of corresponding criteria for the
two-valued case. Notably, one can show that in the Eilenberg-Moore
setting, one essentially always has a characteristic modal
logic~\cite{DBLP:journals/corr/abs-2307-14826}, in sharp contrast to
our present negative result. The flip side of the coin is that
Eilenberg-Moore style trace semantics applies to only rather few
system types (essentially automata with effects), and in particular
does not support a metric on the labels as found, for instance, in
standard metric transition systems.

Our second, now positive, contribution in the present work is to
extend the framework to unrestricted graded semantics, notably
including Kleisli-style coalgebraic trace semantics and, hence, trace
semantics on systems with labels taken from a metric space. For the
syntactic treatment of spectra of behavioural distances in this sense,
we introduce a graded extension of \emph{quantitative
  algebra}~\cite{mardare2016quantitative} that allows describing
graded monads on the category of metric spaces by operations and
approximate equations. As suggested by our negative result,
establishing expressivity of graded logics in the general case
presents additional challenges compared to the two-valued variant and
the Eilenberg-Moore case. In particular, it turns out that the
expressivity criterion needs to be parametric in a strengthening of
the inductive hypothesis in the induction on depth of look-ahead that
it encapsulates; indeed, this happens already in strikingly simple
cases such as metric streams. We develop a number of example
applications: We recover results on expressivity of quantitative modal
logics for (finite-depth) branching-time
distances%
~\cite{DBLP:conf/concur/KonigM18,DBLP:conf/concur/WildS20,DBLP:conf/lics/FordMS21,DBLP:conf/lics/KomoridaKKRH21},
as well as a recent result on expressivity of a quantitative modal
logic for trace distance in metric transition
systems~\cite{DBLP:conf/csl/BeoharG0M23}, which in fact we generalize
to systems with metric state space and closed branching. In a
concluding case study, we moreover identify a new characteristic modal
logic for trace distance on fuzzy metric transition systems, which
turns out to require next-step modalities incorporating a constant
shift on label distances.

Omitted proofs and additional details can be found in the appendix.

\subparagraph*{Related Work}
We have mentioned previous work on coalgebraic branching-time behavioural
distances~\cite{bbkk:trace-metrics-functor-lifting,DBLP:conf/concur/KonigM18,DBLP:conf/csl/Forster0HNSW23,DBLP:conf/fossacs/WildS21,DBLP:journals/lmcs/WildS22,DBLP:conf/csl/BeoharG0M23,DBLP:conf/lics/KomoridaKKRH21}
and on graded semantics for two-valued behavioural equivalences and
preorders~\cite{DBLP:conf/calco/MiliusPS15,DBLP:conf/concur/DorschMS19,DBLP:conf/lics/FordMS21}. Kupke
and Rot~\cite{DBLP:journals/lmcs/KupkeR21} study logics for
\emph{coinductive predicates}, which generalize branching-time
behavioural distances. Generally, our overall setup differs from the
one used in~\cite{DBLP:journals/lmcs/KupkeR21} and elsewhere by
working with coalgebras that already live on metric spaces
(e.g.~\cite{Rutten98,Worrel00,bw:behavioural-pseudometric,DBLP:conf/csl/Forster0HNSW23,GoncharovEA23});
this allows covering functors on metric spaces that are not liftings
of set functors, such as the full Hausdorff functor (which takes
closed subsets). Recent work on Galois connections for logical
distances~\cite{DBLP:conf/csl/BeoharG0M23,DBLP:conf/stacs/BeoharG0MFSW24}
is highly general (and in fact not even tied to state-based systems)
but leaves more work to the instantiation than the framework of
graded monads. Moreover, it is aimed primarily at fixpoint
characterizations of logical distance, and in fact induces behavioural
distance from the logic, while we aim to provide logical
characterizations of \emph{given} behavioural distances. %
Alternative coalgebraic approaches to process equivalences coarser
than branching time include coalgebraic trace semantics in
Kleisli~\cite{DBLP:journals/lmcs/HasuoJS07} and Eilenberg-Moore
categories~\cite{DBLP:journals/jcss/Jacobs0S15}, which are both
subsumed by the paradigm of graded
monads~\cite{DBLP:conf/calco/MiliusPS15}, as well as an approach in
which behavioural equivalences are \emph{defined} via characteristic
logics~\cite{DBLP:journals/corr/KlinR16}. The Eilenberg-Moore and
Kleisli setups can be unified using corecursive algebras, which also
support, under certain assumptions, a logical characterization for
these cases~\cite{DBLP:journals/logcom/RotJL21}. The Eilenberg-Moore
approach has been applied to linear-time behavioural
distances~\cite{bbkk:trace-metrics-functor-lifting}.  Recently, some
of the present authors used the
graded-semantics approach to Eilenberg-Moore semantics to extract
characteristic logics that factor through the determinization of a
coalgebra~\cite{DBLP:journals/corr/abs-2307-14826}.  We make use of their notion of \emph{graded logic} and
complement their work by considering unrestricted graded semantics, in
particular covering the more broadly applicable Kleisli-style
semantics.

De Alfaro et al.~\cite{afs:linear-branching-metrics} introduce a
linear-time logic for (state-labelled) metric transition systems. %
The semantics of this logic is defined by first computing the set of
paths of a system, so that propositional operators and modalities have
a different meaning than in corresponding branching-time logics, while
our graded logics are fragments of branching-time logics. Fahrenberg
et al.~\cite{FahrenbergEA11} present a game-based approach to a
spectrum of behavioural distances on metric transition systems.  A
two-valued logic for probabilistic trace semantics (for a discrete set
of labels) has been considered in the context of differential
privacy~\cite{ccp:logic-differential-privacy}. A notion of logical
distance is then obtained via a real-valued semantics defined using a
syntactic distance on formulae; this semantics is not compositional
(truth values are defined by taking infima over the whole logical
syntax), so subsequent results relating this logical distance to
notions of weak anonymity do not contradict our impossibility result on
(compositional) characteristic logics for probabilistic trace
semantics.

\section{Preliminaries}\label{sec:prelims}
Basic familiarity with category theory is assumed
(e.g.~\cite{DBLP:books/daglib/0023249}). We write~$\SET$ for the
category of sets and maps. Below, we recall some background on
(bounded) metric spaces and universal coalgebra.

\subparagraph*{Metric spaces}
\noindent The real unit interval $[0,1]$ will serve as the domain of
distances and truth values. Under the usual ordering $\le$, $[0,1]$
forms a complete lattice; we write $\bigvee,\bigwedge$ for joins and
meets in $[0,1]$ (e.g.\ $\bigvee_i x_i=\sup_i x_i$), and $\vee,\wedge$
for binary join and meet, respectively. We denote truncated addition
and subtraction by $\oplus$ and~$\ominus$, respectively;
that is, $x\oplus y=\min(x+y,1)$ and $x\ominus y=\max(x-y,0)$. These
operations form part of a structure of $[0,1]$ as a (co-)quantale; for
readability, we refrain from working with more general
quantales~\cite{DBLP:conf/fossacs/WildS21,DBLP:conf/csl/Forster0HNSW23}.

\begin{defn}A (bounded) \emph{pseudometric space} is a pair $(X, d)$
  consisting of a set $X$ and a function $d\colon X\times X \to [0,1]$
  satisfying the standard conditions of \emph{reflexivity} ($d(x,x)=0$
  for all $x\in X$), \emph{symmetry} ($d(x,y)=d(y,x)$ for all
  $x,y\in X$), and the \emph{triangle inequality}
  ($d(x, z)\le d(x,y) + d(y, z)$ for all $x, y, z \in X$); if
  additionally \emph{separation} holds (for $x,y \in X$, if
  $d(x, y) = 0$ then $x = y$), then~$(X,d)$ is a \emph{metric
    space}. A function $f\colon X \to Y$ between pseudometric spaces
  $(X, d_X)$ and $(Y, d_Y)$ is \emph{nonexpansive} if
  $d_Y(f(x), f(y)) \leq d_X(x, y)$ for all $x, y \in X$.  Metric
  spaces and nonexpansive maps form a category $\MET$.
\end{defn}
\noindent We often do not distinguish notationally between a
(pseudo-)metric space $(X,d)$ and its underlying set~$X$. Occasionally
we
use subscripts to make explicit the carrier to which a (pseudo-)metric is
associated, i.e. $d_X$ is the
(pseudo-)metric of the space with carrier $X$. The categorical
product $(X,d_X) \times (Y, d_Y)$ of (pseudo-)metric spaces equips
the Cartesian product $X \times Y$ with the supremum (pseudo-)metric
$d_{X \times Y}((a, b), (a', b')) = d_X(a, a') \vee d_Y(b,
b')$. Similarly, the \emph{Manhattan tensor} $\boxplus$ equips
$X\times Y$ with the \emph{Manhattan (pseudo-)metric}
$d_{X\boxplus Y}((a,b), (a',b')) = d_X(a, a') \oplus d_Y(b, b')$.  We
occasionally write elements of the product $X^{n+m}$ as $vw$ if
$v\in X^n$ and $w \in X^m$. Given
(pseudo-)metric spaces $X,Y$, the nonexpansive functions
$X\to Y$ form a (pseudo-)metric space under the standard supremum
distance.

\begin{example}
    \label{expl:monads}
    We recall some key examples of functors on~$\SET$ and~$\Met$.
  \begin{enumerate}[wide]
  \item We write $\Pset$ for the finite powerset functor on~$\SET$,
    and~$\Pmet$ for the lifting of $\Pset$ to $\MET$ given by the
    Hausdorff metric. Explicitly, for a metric space $(X, d)$ and
    $A, B \in \Pset X$,
    \begin{equation}\label{eq:hausdorff}
      d_{\Pmet X}(A, B) =  \textstyle
      (\bigvee_{\raisebox{1pt}{$\scriptstyle a\in A$}} \bigwedge_{b\in B} d(a, b)) \vee %
      (\bigvee_{\raisebox{1pt}{$\scriptstyle{b\in B}$}}\bigwedge_{\raisebox{1pt}{$\scriptstyle a\in A$}}d(a, b)).
    \end{equation}
    Both~$\Pset$ and~$\Pmet$ are monads, with multiplication taking
    big unions.
\item Related to the above, the closed Hausdorff monad $\Pc$ on $\MET$ sends a metric space~$X$ to the set of closed subsets of $X$, again equipped with the Hausdorff metric. For a nonexpansive function $f\colon X \to Y$,  $\Pc f$ sends  $A \in\Pc X$ to the closure of $f[A]$. Monad multiplication takes the closure of the big union.
\item Similarly, $\ftDist$ denotes the functor on $\SET$ that maps a
  set $X$ to the set of finitely supported probability distributions
  on $X$, and $\ftFinK$ denotes the lifting of $\ftDist$ to $\MET$
  that equips $\ftDist X$ with the Kantorovich metric. Explicitly, for
  a metric space $(X,d)$ and $\mu, \nu \in \ftDist X$,
    \begin{equation*}
      \textstyle d_{\ftFinK X}(\mu, \nu) = \bigvee_f \sum_{x\in X} f(x)
      (\mu(x) \ominus \nu(x))
    \end{equation*}
     where $f$
    ranges over all nonexpansive functions $X \to [0, 1]$. We often
    write elements of $\ftFinK X$ as finite formal sums
    $\sum p_i\cdot x_i$, with $x_i\in X$ and $\sum p_i=1$.
  \item\label{item:fuzzy-power} The \emph{finite fuzzy powerset}
    functor $\FPow_\omega$ is given on sets~$X$ by
    $\FPow_\omega X=\{A\colon X\to [0,1]\mid A(x)=0\text{ for almost
      all $x\in X$}\}$, and on maps $f\colon X\to Y$
    by %
    $\FPow_\omega f(A)(y)=\bigvee\{A(x)\mid f(x)=y\}$ for
    $A\in\FPow_\omega X$. That is, $\FPow_\omega X$ consists of the
    finite fuzzy subsets of~$X$, given by assigning membership degrees
    in $[0,1]$ to elements of~$X$, and $\FPow_\omega f$ acts by taking
    fuzzy direct images. We lift~$\FPow_\omega$ to a
    functor~$\FHaus_\omega$ on metric
    spaces %
  that equips $\FPow_\omega X$ with the fuzzy Hausdorff
  distance~\cite[Example~5.3.1]{DBLP:conf/fossacs/WildS21}. Explicitly,
		  $d_{\FHaus_\omega X}(A,B)=d_0(A,B)\vee d_0(B,A)$ for a metric space $(X,d)$ and
  $A,B\in\FPow_\omega X$,  where
  \begin{equation*}\label{eq:fuzzy-haus-main}\textstyle
    d_0(A,B)=\bigvee_{x}\bigwedge_{y} (A(x)\ominus B(y))\vee (A(x)\wedge d(x,y) ). %
  \end{equation*}
  Thus, $d_0(A,B)$ is analogous to the left-hand term in the binary
  join defining the Hausdorff metric~\eqref{eq:hausdorff}: Both terms
  can be read intuitively as `$B$ is far from~$A$ if there is~$x$ such
  that $x\in A$ and for all $y$, if $y\in B$ then~$y$ is far
  from~$x$', where~$d_0(A,B)$ takes into account that the sets~$A,B$
  are fuzzy (in particular, the `if $y\in B$' is reflected in the
  contribution of~$B(y)$ being negative).
  \end{enumerate}
\end{example}

\subparagraph*{Coalgebra} \label{sec:coalg}\noindent\emph{Universal
  coalgebra}~\cite{DBLP:journals/tcs/Rutten00} has established itself
as a way to reason about state-based systems at an appropriate level
of abstraction. It is based on encapsulating the transition type of
systems as an endofunctor $G\colon \catC \to \catC$ on a base
category~$\catC$. Then, a $G$-coalgebra $(X, \gamma)$ consists of a
$\catC$-object $X$, thought of as an object of \emph{states}, and a
morphism $\gamma\colon X \to G X$, thought of as assigning to each
state a collection of successors, structured according to~$G$. A
$\catC$-morphism $h\colon X \to Y$ is a morphism
 of $G$-coalgebras $(X, \gamma)\to(Y, \delta)$ if
 $G h \cdot \gamma = \delta \cdot h$.

 For a functor $G\colon \MET \to \MET$, one has a canonical notion of
 \emph{branching-time behavioural distance} $d^{G}_\gamma$ on a
 $G$-coalgebra $(X,\gamma)$~\cite{DBLP:conf/csl/Forster0HNSW23}. In
 case~$G$ is a lifting of a set functor (which means roughly that the
 underlying set of $GX$ is independent of the metric on~$X$), the
 general definition simplifies as follows: $d^G_\gamma$ is the least
 fixpoint of the map
 $d\mapsto
 d_{G(X,d)}\circ(\gamma\times\gamma)$~\cite{bbkk:trace-metrics-functor-lifting,DBLP:conf/csl/Forster0HNSW23}.

\begin{example}\label{expl:coalg}
  Throughout the paper, we \emph{fix a metric space~$\Act$} of
  \emph{labels}. Finitely branching metric transition systems with
  transition labels in~$\Act$ are coalgebras for the functor
  $\Pmet(\Act\times{-})$. (More precisely, a metric transition system
  is usually assumed to have a set as its state space, while
  $\Pmet(\Act\times{-})$-coalgebras more generally have a metric space
  of states, subsuming mere sets of states as discrete metric
  spaces). Similarly, coalgebras for the functor $\Pc(\Act\times{-})$
  are \emph{closed-branching} metric transition systems, where sets of
  successors can be infinite but are required to be closed. With few
  exceptions (e.g.~\cite{DBLP:conf/csl/Forster0HNSW23}), most
  coalgebraic approaches to behavioural metrics (e.g.\
  \cite{bbkk:trace-metrics-functor-lifting,DBLP:conf/concur/KonigM18,DBLP:journals/lmcs/WildS22,DBLP:journals/lmcs/KupkeR21})
  rely on the functor being a lifting of a $\SET$-functor. We work
  with unrestricted functors on $\MET$, thus, e.g., covering the
  above-mentioned functor $\Pc(\Act\times{-})$, which is not a lifting
  of a set functor.  We use trace semantics on metric labelled
  transition systems (both finitely branching and closed-branching) as
  a running example of concepts as they appear throughout the text.
\end{example}

\subparagraph*{Quantitative Coalgebraic Modal Logic} We proceed to
introduce the requisite notion of quantitative coalgebraic modal
logic~\cite{DBLP:conf/ijcai/SchroderP11,DBLP:conf/concur/KonigM18,DBLP:journals/lmcs/WildS22,DBLP:journals/corr/abs-2307-14826},
in a formulation geared towards easing the extraction of invariant
fragments for various
semantics~\cite{DBLP:journals/corr/abs-2307-14826}, and instantiated
to the category of metric spaces. The notion of quantitative
coalgebraic modal logic will also serve as the yardstick for our
negative result on characteristic modal logics for probabilistic
metric trace semantics (\autoref{sec:prob-met-trace}).

Syntactically, a \emph{modal logic} is a triple
$\mathcal{L} = (\Theta, \mathcal{O}, \Lambda)$ where~$\Theta$ is a set
of truth constants,~$\mathcal{O}$ is a set of propositional operators,
each with associated finite arity, and~$\Lambda$ is a set of modal
operators, also each with an associated finite arity. For readability,
we restrict to unary modal operators; extending our positive results
to modal operators of higher arity is simply a matter of adding
indices. The set of \emph{formulae} of $\mathcal{L}$ is then given by
the grammar
\begin{equation*}
  \phi ::= c \mid p(\phi_1, \ldots, \phi_n) \mid L\phi \qquad\qquad(c
  \in \Theta,\; p
  \in \mathcal{O}\text{ $n$-ary},\; L \in \Lambda).
\end{equation*}
Formulae are interpreted in $G$-coalgebras for a given functor
$G\colon \MET \to \MET$, and take values in the truth value object
$\Omega=[0,1]$, which we equip with the standard metric
$d_\Omega(x,y)=|x-y|$.  Moreover, the semantics is parametric in the
following components:

\begin{itemize}
\item For every $c \in \Theta$, a nonexpansive map
  $\hat{c}\colon 1 \to \Omega$.
\item For $p \in \mathcal{O}$ with arity $n$, a nonexpansive map
  $\sem{p}\colon\Omega^n \to \Omega$
	\item For $L\in \Lambda$, a nonexpansive map $\sem{L}\colon G\Omega\to \Omega$
\end{itemize}

\noindent The evaluation of a formula $\phi$ on a $G$-coalgebra
$(X, \gamma)$ is then a nonexpansive map
$\sem{\phi}_\gamma\colon X \to \Omega$, inductively defined
by\jfnote{Possibly layout a bit more nicely}
\begin{equation*}
\begin{gathered}
  \sem{c}_\gamma = (X \xrightarrow{!} 1 \xrightarrow{\hat{c}} \Omega)
  \qquad \qquad \sem{p(\phi_1, \ldots, \phi_n)}_\gamma  =
  (X\xrightarrow{\langle \sem{\phi_1}_\gamma, \ldots,
  \sem{\phi_n}_\gamma\rangle}\Omega^n\xrightarrow{\sem{p}}\Omega)
  \\[0.5em]
  \sem{L \phi}_\gamma  = (X\xrightarrow{\gamma} GX\xrightarrow{G\sem{\phi}_\gamma} G\Omega\xrightarrow{\sem{L}}\Omega)
\end{gathered}
\end{equation*}
\begin{example}\label{expl:modalities}
  We briefly exemplify the semantics of modalities: Take the functor
  $G=\Pmet(\Act\times (-))$ modelling metric transition systems
  (\autoref{expl:coalg}), and define the interpretation
  $\sem{\modal{a}}\colon\Pmet(\Act\times\Omega)\to\Omega$ of
  modalities $\Diamond_a$, for $a\in\Act$, by
  $\sem{\modal{a}}(U) = \bigvee_{(b, v)\in U} (1- d(a,b))\land
  v$. Then, roughly speaking, the degree to which a state in a metric
  transition system satisfies a formula $\Diamond_a\phi$ is the degree
  to which it has a $b$-successor that satisfies~$\phi$, for some~$b$
  that is close to~$a$. (The use of $1- d(a,b)$ is owed to the usual
  discrepancy between~$1$ representing `true' but also `far apart'.)
\end{example}
\noindent In the framework defined so far, truth constants are
interchangeable with nullary propositional operators, but in the
setting of graded logics (\autoref{sec:logic}), the two concepts will
play syntactically and semantically distinct roles. In particular,
invariance w.r.t.\ a target semantics (\autoref{thm:invariance}) will
in general hold only for formulae of \emph{uniform depth}, that is,
formulae in which all occurrences of truth constants are nested under
the same number of modal operators. In cases where there are no truth
constants, all formulae are
uniform. %
We write $\mathcal{L}_\unif$ for the set of uniform-depth
$\mathcal L$-formulae.

\begin{defn}
	\label{dfn:log-distance}
        \emph{Logical distance} under the logic $\mathcal{L}$ on a
        $G$-coalgebra $(X, \gamma)$ is the
        pseudometric~$d^{\mathcal{L}}$ given by
  $\textstyle d^{\mathcal{L}}(x, y) = \bigvee \{ d_\Omega(\sem{ \phi}_\gamma(x), \sem{ \phi }_\gamma(y)) \mid \phi \in
    \mathcal{L}_\unif\}$.\jfnote{Expand when space is available}
\end{defn}

\noindent Logical distance is always a lower bound for branching-time
behavioural
distance~\cite{DBLP:conf/concur/KonigM18,DBLP:journals/lmcs/WildS22,DBLP:conf/csl/Forster0HNSW23};
we discuss details in \autoref{expl:branching-time}.

\section{Probabilistic Metric Trace
Semantics}\label{sec:prob-met-trace} Finitely branching
\emph{probabilistic metric transition systems} over a metric space of
transition labels~$\Act$ are coalgebras for the functor
$\Gprob = \ftFinK(\Act\boxplus(-))$ (cf.\ \autorefexpls{expl:monads}
and~\ref{expl:coalg}). %
  The \emph{probabilistic (metric) trace
    semantics}~\cite{DBLP:conf/concur/Christoff90} of a probabilistic
  transition system calculates, at each depth~$n$, a distribution over
  length-$n$ traces. One then obtains a notion of \emph{depth-$n$
    probabilistic trace distance} $d_n^\ptrace$, which takes
  Kantorovich distances of depth-$n$ trace distributions under the
  Manhattan distance on traces. Formal definitions are as follows.
  \begin{defn}
    We write $\Act^{\boxplus n}$ for the $n$-fold Manhattan tensor
    $\Act\boxplus\dots\boxplus\Act$. Let $(X, \gamma)$ be a
    $\Gprob$-coalgebra.  For each $x\in X$, the \emph{depth-$n$ trace
      distribution} $\mu^n_x \in \ftFinK(\Act^{\boxplus n})$ is
    inductively defined as
    $\mu^{n+1}_x(aw) = \sum_{y\in X}\gamma(x)(a,y)\mu^{n}_y(w)$ for
    $a \in \Act$ and $w\in \Act^n$, with
    $\mu^0_x\in \ftFinK(\Act^{\boxplus 0}) \cong \ftFinK(1)$ being the
    unique distribution on the singleton set~$1$.  The
    \emph{probabilistic trace distance} on $X$ is
    $d^\ptrace=\bigvee_{n<\omega}d_n^\ptrace$, where
    $d^\ptrace_n(x, y) = d_{\ftFinK(\Act^{\boxplus n})}(\mu^n_x,
    \mu^n_y)$.
  \end{defn}

\noindent Consider the following concrete example, where we assume
that $d(b,c) = 0.5$.

\begin{tikzpicture}[node distance=2cm, auto]
  \node at (0, 0) (A)[circle, fill,inner sep=2pt] {};
  \node at ([shift={(160:0.4)}]A) {$x$};
  \node at (1.5, 0) (E)[circle, fill,inner sep=2pt] {};

  \node at (7, 0) (A1)[circle, fill,inner sep=2pt] {};
  \node at ([shift={(45:0.4)}]A1) {$y$};
  \node at (5.5, 1) (D1)[circle, fill,inner sep=2pt] {};
  \node at (5.5, -1) (D2)[circle, fill,inner sep=2pt] {};

  \node at (4.0, 0) (Sink)[circle, fill,inner sep=2pt] {};

  \path[->,arrows = {-Stealth[length=6pt, inset=2pt]}]
    (A) edge node[below] {\(1\)} node[above] {$a$} (E)

    (Sink) edge[out=220,in=320,looseness=20]  node[below] {\(1\)} node[above] {$a$} (Sink)

    (D1) edge node[above left=-2pt] {$b$} node[below right=-2pt] {$1$}(Sink)
    (D2) edge node[above right=-2pt] {$c$} node[below left=-2pt] {$1$}(Sink)

    (E) edge[out=30, in=150] node[above] {$b\quad\frac{1}{2}$} (Sink)
    (E) edge[out=-30, in=-150] node[below] {$c\quad\frac{1}{2}$} (Sink)

    (A1) edge node[above right=-2pt] {\(a\)} node [below left=-2pt] {$\frac{1}{4}$} (D1)
    (A1) edge node[above left=-2pt] {\(a \)} node[below right=-2pt]{$\frac{3}{4}$} (D2);

\end{tikzpicture}

\noindent For  $n \ge 2$ we then have by the above definition that
$\mu^{n}_x = \frac{1}{2}aba^{n-2} + \frac{1}{2}aca^{n-2}$ while
$\mu^{n}_y = \frac{1}{4}aba^{n-2} + \frac{3}{4}aca^{n-2}$. The
distance of the two relevant traces is given by $d(aba^{n-2}, aca^{n-2}) = d(a, a)
\oplus d(b,c) \oplus d(a,a) \oplus \ldots \oplus d(a,a)= 0.5$.
Calculating the Kantorovich distance of trace distributions then gives us that $d(\mu^{n}_x,
\mu^{n}_y) = \frac{1}{4}d(aba^{n-2}, aca^{n-2}) = 0.125$, and by
extension $d^\ptrace(x,y) = 0.125$.

  One would now like to have a logic that characterizes the trace
  distance $d^\ptrace$. However, we establish the following
  impossibility result instead:

\begin{theorem}
  \label{thm:propNoExpr}
  Let $\mathcal{L} = (\Theta, \mathcal{O}, \Lambda)$ be a coalgebraic modal logic
  with unary modalities for the functor
  $\Gprob$, over a  non-discrete metric space~$\Act$ of labels. Then $d^\mathcal{L}\not = d^{\mathsf{ptrace}}$.
\end{theorem}
In other words, no quantitative coalgebraic modal logic with unary
modalities has a compositionally defined fragment that characterizes
probabilistic metric trace distance. The restriction to coalgebraic
modal logics effectively means only that modal logics should be
invariant under the standard branching-time semantics and have only
next-step
modalities~\cite{DBLP:journals/ndjfl/Pattinson04,DBLP:journals/tcs/Schroder08}. \autoref{thm:propNoExpr}
implies in particular that the logic featuring modalities $\Diamond_a$
for $a\in\Act$, with $\Diamond_a\phi$ being the expected truth value
of~$\phi$ restricted to $a$-successors, fails to characterize
probabilistic metric trace distance (even though it characterizes
two-valued probabilistic trace
\emph{equivalence}~\cite{DBLP:journals/mscs/BernardoB08,DBLP:conf/concur/DorschMS19}).
In fact, it can even be shown that giving up the
requirement of interpretations of modalities being nonexpansive does
not help.

\begin{proof}[Proof sketch (\autoref{thm:propNoExpr})]
  Suppose that $\mathcal{L}$ is invariant under probabilistic metric
  trace semantics ($d^{\mathcal{L}}\le d^{\mathsf{ptrace}}$); we show
  that $\mathcal{L}$ fails to be expressive
  ($d^{\mathcal{L}}\not\ge d^{\mathsf{ptrace}}$). As an intermediate
  step, we show that invariance under probabilistic metric trace
  semantics implies that modal operators are affine maps. Then
  calculation shows that affine modalities are unable to distinguish
  the states~$x$ and~$y$ in the following system, where
  $d(a,b) = v<1$, to a degree greater than~$v^2$, even though the
  behavioural distance of~$x$ and~$y$ under probabilistic trace
  semantics is~$v$.

\vspace{-1.3em}  \centering
\qquad\qquad\qquad\qquad
\begin{tikzpicture}[node distance=2cm, auto]
				\node at (2.5, 1) (A)[circle, fill,inner sep=2pt] {};
				\node at ([shift={(45:0.4)}]A) {$x$};
				\node at (1.5, 0) (D)[circle, fill,inner sep=2pt] {};
				\node at ([shift={(160:0.4)}]D) {$x_a$};
				\node at (3.5, 0) (E)[circle, fill,inner sep=2pt] {};
				\node at ([shift={(30:0.4)}]E) {$x_b$};

				\node at (6.5, 1) (A1)[circle, fill,inner sep=2pt] {};
				\node at ([shift={(45:0.4)}]A1) {$y$};
				\node at (5.5, 0) (D1)[circle, fill,inner sep=2pt] {};
				\node at ([shift={(160:0.4)}]D1) {$y_b$};
				\node at (7.5, 0) (E1)[circle, fill,inner sep=2pt] {};
				\node at ([shift={(30:0.4)}]E1) {$y_a$};

				\path[->,arrows = {-Stealth[length=6pt, inset=2pt]}]
					(A) edge node[left=2ex] {\(\frac{1}{2}\)} node[right=0.5ex] {$a$} (D)
						edge node[right=2ex] {\(\frac{1}{2}\)} node[left=0.5ex] {$b$} (E)

					(D) edge[out=220,in=320,looseness=30]  node[below] {\(1\)} node[above] {$a$} (D)
					(E) edge[out=220,in=320,looseness=30]  node[below] {\(1\)} node[above] {$b$} (E)

					(D1) edge[out=220,in=320,looseness=30]  node[below] {\(1\)} node[above] {$b$} (D1)
					(E1) edge[out=220,in=320,looseness=30]  node[below] {\(1\)} node[above] {$a$} (E1)

					(A1) edge node[left=2ex] {\(\frac{1}{2}\)} node[right=0.5ex] {$a$}(D1)
						 edge node[right=2ex] {\(\frac{1}{2}\)} node[left=0.5ex] {$b$}(E1) ; 
\end{tikzpicture}\vspace{-0.8em}
\end{proof}

\noindent We leave the question of whether a characteristic logic with higher-arity
modalities exists as an open problem.

While expressive quantitative coalgebraic logics for branching-time
semantics exist for a wide variety of
systems~\cite{DBLP:conf/concur/KonigM18,DBLP:journals/lmcs/WildS22,DBLP:conf/csl/Forster0HNSW23,GoncharovEA23},
this is thus apparently not always the case for linear-time
semantics. The no-go result above emphasizes the challenges of the
quantitative setting and the need for a theory of quantitative
coalgebraic logics beyond branching time. In the following, we will
address precisely this problem, by adopting techniques from the theory
of graded semantics and highlighting issues unique to the metric
setting.

\section{Graded Monads and Graded Algebras} The framework of
\emph{graded
  semantics}~\cite{DBLP:conf/concur/DorschMS19,DBLP:conf/calco/MiliusPS15}
is based on the central notion of \emph{graded monads}, which
algebraically describe the structure of observable behaviours, in
particular identifications beyond branching time, at each finite
depth. Here, \emph{depth} is understood as look-ahead, measured in
terms of the number of transition steps.

\begin{defn}
  A \emph{graded monad}
  $\mathbb{M}=((M_n)_{n\in\nat},\eta,(\mu^{n,k})_{n,k\in\nat})$ on a
  category\;$\catC$ consists of a family of functors
  $M_n\colon \catC \to \catC$ for $n \in \nat$ and natural
  transformations $\eta\colon \Id \to M_0$ (the \emph{unit}) and
  $\mu^{n,k}\colon M_nM_k \to M_{n+k}$ for all $n, k \in \nat$ (the
  \emph{multiplications}), subject to essentially the same laws as
  ordinary monads up to the insertion of grades; specifically, one has
  \emph{unit laws} $\mu^{0,n}\cdot \eta M_n=\id_{M_n}=\mu^{n,0}\cdot M_n\eta$
  and an \emph{associative law}
  $\mu^{n+k,m} \cdot \mu^{n,k}M_m=\mu^{n,k+m}\cdot M_n\mu^{k,m}$.
\end{defn}
In particular, $(M_0, \eta, \mu^{00})$ is an ordinary (non-graded)
monad.

The understanding of the data constituting a graded monad is similar
as for plain monads: Roughly speaking (this will be made more precise
in \autoref{sec:theories}), $M_nX$ may be thought of as a space of
terms of depth~$n$, modulo given identities, over variables from~$X$;
$\mu^{nk}$ substitutes depth-$k$ terms into a depth-$n$ term,
obtaining a depth-$(n+k)$ term; and~$\eta$ converts variables into terms
of depth~$0$.
\begin{example}\label{expl:graded-monads}
  We discuss graded monads modelling the linear-time end of the
  spectrum, noting that graded monads cover also branching-time
  (\autoref{expl:branching-time}) and intermediate semantics,
  involving simulation, readiness, failures
  etc.~\cite{DBLP:conf/concur/DorschMS19}.  A \emph{Kleisli
    distributive law} is a natural transformation
  $\lambda\colon FT \to TF$ where~$F$ is a functor and~$T$ a monad,
  subject to coherence with the monad
  structure~\cite{DBLP:journals/lmcs/HasuoJS07}. This yields a graded
  monad with $M_n = TF^n$~\cite{DBLP:conf/calco/MiliusPS15}; here,~$T$
  may be understood as defining the branching type of the system,
  and~$F$ as defining a type of accepted structure. We will use the
  following instance of this construction as a running example: Take
  $F = \Act \times (-)$ and $T = \Pmet$ or $T = \Pc$ (corresponding to
  nondeterministic branching). Then
  $\lambda(a, U) = \{(a, x) \mid x \in U\}$ defines a distributive law
  $\lambda\colon \Act \times T({-}) \to T(\Act \times ({-}))$ (in
  particular,~$\lambda$ is nonexpansive). We obtain the \emph{graded
    metric trace monads} $M_n =T(\Act^n \times (-))$.  \end{example}

\noindent Graded monads come with a graded analogue of Eilenberg-Moore
algebras, which play a central role in the semantics of graded
logics~\cite{DBLP:conf/calco/MiliusPS15,DBLP:conf/concur/DorschMS19}.

\begin{defn}[Graded Algebra]
  Let $\mathbb{M}$ be a graded monad in $\catC$. A \emph{graded
  $M_n$-algebra} $((A_k)_{k \leq n}, (a^{mk})_{m+k \leq n})$ consists
  of a family of $\catC$-objects $A_i$ and morphisms
  $a^{mk}\colon M_mA_k \to A_{m+k}$ satisfying essentially the same
  laws as a monad algebra, up to insertion of the
  grades. Specifically, we have $a^{0m} \cdot \eta_{A_m} = \id_{A_m}$ for
  $m \leq n$, and whenever $m + r + k \leq n$, then
  $a^{m+r,k} \cdot \mu^{m,r}_{A_k}=a^{m,r+k} \cdot M_ma^{r,k}$.
  An \emph{$M_n$-homomorphism} of $M_n$-algebras $A$ and $B$ is a
  family $(f_k\colon A_k \to B_k)_{k \leq n}$ of maps such that
  whenever $m + k \leq n$, then
  $f_{m+k} \cdot a^{m,k}=b^{m,k} \cdot M_mf_k$. Graded $M_n$-Algebras
  and their homomorphisms form a category $\Galg{n}{\mathbb{M}}$.
\end{defn}

\noindent That is, elements of a graded algebra are stratified by
depth, and applying an operation of depth~$m$ to elements of depth~$k$
yields elements of depth $m+k$, For $n = 1$, this definition
instantiates as follows: An $M_1$-algebra is a tuple
$(A_0, A_1,a^{00},a^{01},$ $ a^{10})$, such that 1) $(A_0, a^{00})$
and $(A_1, a^{01})$ are $M_0$-algebras.  2) (Homomorphy)
$a^{10} \colon M_1A_0 \to A_1$ is an $M_0$-homomorphism
$(M_1A_0, \mu^{01})\to(A_1, a^{01})$.  3) (Coequalization)
$a^{10} \cdot M_1a^{00} = a^{10} \cdot \mu^{10}$, i.e. the following
diagram commutes (without necessarily being a coequalizer):
	\begin{equation}\label{eq:coeq}
	\begin{tikzcd}
		M_1M_0A_0 \arrow[r,swap, "\mu^{10}", shift right] \arrow[r, "M_1a^{00}", shift left] & M_1A_0 \arrow[r, "a^{10}"] & A_1
	\end{tikzcd}
	\end{equation}

\noindent It is easy to see that $((M_kX)_{k\leq n}, (\mu^{m,k})_{m+k \leq n})$
is an $M_n$-algebra for every $\catC$-object $X$.
 Again, $M_0$-algebras are just
(non-graded) algebras for the monad $(M_0, \eta, \mu^{00})$.
        \noindent The semantics of modalities will later need the
        following property:

\begin{defn}[Canonical algebras]
  Let
  $({-})_0\colon \Galg{1}{\mathbb{M}}\rightarrow \Galg{0}{\mathbb{M}}$ be
  the functor taking an $M_1$-algebra
  $A=((A_k)_{k \leq 1}, (a^{mk})_{m+k \leq 1})$ to the $M_0$-algebra
  $(A_0, a^{00})$. An $M_1$-algebra $A$ is
  \emph{canonical} if it is free over $(A)_0$, i.e.\ if for every
  $M_1$-algebra $B$ and $M_0$-homomorphism
  $f: (A)_0 \rightarrow (B)_0$, there is a unique $M_1$-homomorphism
  $g: A \rightarrow B$ such that $(g)_0 = f$.
\end{defn}

\begin{lemma}(\cite[Lemma 5.3]{DBLP:conf/concur/DorschMS19})
  \label{lem:canonical}
  An $M_1$-algebra~$A$ is canonical iff~\eqref{eq:coeq} is a
  coequalizer diagram in the category of $M_0$-algebras.
\end{lemma}

\section{Graded Quantitative Theories}\label{sec:theories}
\lsnote{@Jonas: Maybe swap with graded semantics section}
\noindent Monads on $\SET$ are induced by equational theories
\cite{Linton86}. By equipping each operation with an assigned depth
and requiring each axiom to be of uniform depth, one obtains a notion
of \emph{graded equational theory} which, modulo size issues, can be
brought into bijective correspondence with graded
monads~\cite{DBLP:conf/calco/MiliusPS15}. On the other hand, Mardare
et al.~\cite{mardare2016quantitative} introduce a system of
quantitative equational reasoning, with formulae of the form
$s =_\epsilon t$ understood as ``$s$ differs from $t$ by at most
$\epsilon$''. These quantitative equational theories induce monads on
the category of metric spaces. We introduce a graded version of this
system to present graded monads in $\MET$, keeping to finitary
operations (and hence finite branching) for ease of presentation.

\begin{defn}[Graded signatures, uniform terms]\label{def:terms}
  A \emph{graded signature} consists of an algebraic
  signature~$\Sigma$ and a function
  $\delta\colon \Sigma \rightarrow \nat$ assigning a \emph{depth} to
  each algebraic operation. \emph{Uniform depth} of terms is then
  defined inductively: Variables have uniform depth~$0$, and for
  $m$-ary $f\in\Sigma$, $f(t_1, \ldots, t_m)$ has uniform depth $n+k$
  if $\delta(f) = n$ and all~$t_i$ have uniform depth~$k$. In
  particular, constants $c \in \Sigma$, as terms, have uniform
  depth~$n$ for all $n \geq \delta(c)$.  We write
  $\mathbb{T}^\Sigma_nX$, or just $\mathbb{T}_nX$, for the set of
  terms of uniform depth $n$
  over~$X$. %
  A \emph{substitution of uniform depth $n$} is a function
  $\sigma\colon X \rightarrow \mathbb{T}_nY$. Such a substitution
  extends to a map
  $\sigma: \mathbb{T}_kX \rightarrow \mathbb{T}_{k+n}Y$ on terms for
  all $k \in \nat$, where as usual one defines
  $\sigma(f(t_1, \ldots, t_m)) = f(\sigma(t_1), \ldots,
  \sigma(t_m))$. A substitution is \emph{uniform-depth} if it is of
  uniform depth~$n$ for some~$n$.
\end{defn}

\begin{defn}[Graded quantitative theory]\label{def:graded-theory}
  For a set~$Z$, we let $\qeq(Z)$ denote the set of quantitative
  equalities $z_1 =_\epsilon z_2$ where $z_1,z_2 \in Z$ and
  $\epsilon \in \V$.
  Given a set~$X$ of variables, we then write
  $\qeq(\mathbb{T}(X)) = \bigcup_{n\in \nat} \qeq(\mathbb{T}_n(X))$;
  that is, $\qeq(\mathbb{T}(X))$ is the set of uniform-depth
  quantitative equalities among $\Sigma$-terms over~$X$.  A
  \emph{quantitative theory} $\mathcal{T} = (\Sigma, \delta, E)$
  consists of a graded signature $(\Sigma, \delta)$ and a set
  $E \subseteq \mathcal{P}(\qeq(X)) \times \qeq(\mathbb{T}X)$ of
  \emph{axioms}.  Axioms $(\Gamma,s=_\epsilon t)$ are written in the
  form $\Gamma\vdash s=_\epsilon t$; we refer to~$\Gamma$ as the
  \emph{context} of the axiom. The \emph{depth} of
  $\Gamma\vdash s=_\epsilon t$ is that of $s=_\epsilon t$.  We say
  that~$\mathcal{T}$ is \emph{depth-1} if all its operations and
  axioms have depth at most~$1$.
\end{defn}
The context $\Gamma$ of an axiom $\Gamma\vdash s=_\epsilon t$ forms a
constraint on the variables that is required in order for
$s=_\epsilon t$ to hold. Correspondingly, \emph{derivability} of
quantitative equalities in~$\qeq(\mathbb{T}(X))$ over a graded
quantitative theory~$\mathcal{T} = (\Sigma, \delta, E)$ in a
\emph{context} $\Gamma_0 \in \mathcal{P}(\qeq(X))$ is defined
inductively by the following rules:
\begin{gather*}
  \lrule{triang}{t=_\epsilon s\quad
  s=_{\epsilon'}u}{t=_{\epsilon+\epsilon'}u} \qquad
  \lrule{refl}{}{s=_0s} \qquad \lrule{sym}{t=_\epsilon
  s}{s=_{\epsilon}t}\\[0.5em]
  \lrule{wk}{t=_\epsilon s}{t=_{\epsilon'}s}\;(\epsilon' \geq \epsilon) \qquad
  \lrule{arch}{\{t =_{\epsilon'}s\mid \epsilon' > \epsilon\}}{t=_\epsilon s} \qquad
  \lrule{assn}{}{\,\phi\,}\;(\phi\in\Gamma_0)\\[0.5em]
\lrule{ax}{\{\sigma(u) \mid u \in \Gamma\}}{\sigma(t)=_\epsilon
  \sigma(s)}\;((\Gamma, t=_\epsilon s)\in E)\qquad
	\lrule{nexp}{t_1=_\epsilon s_1 \qquad\dots\qquad t_n=_\epsilon s_n}
	{f(t_1,\dots,t_n)=_\epsilon f(s_1,\dots,s_n)}
\end{gather*}
where $\sigma$ is a uniform-depth substitution. Note the difference
between rules $(\mathbf{ax})$ and $(\mathbf{assn})$: Quantitative
equalities from the theory can be substituted into, while this is not
sound for quantitative equalities from the context. A graded
quantitative equational theory \emph{presents} a graded monad~$\monad$
on $\MET$ where $M_nX$ is the set of terms of uniform depth $n$ over
variables in~$X$, quotiented by the equivalence relation that
identifies terms $s,t$ if $s=_0t$ is derivable in context~$X$, with
the distance $d_{M_n}([s], [t]) = \epsilon$ of equivalence classes
$[s], [t] \in M_nX$ being the least~$\epsilon$ such that
$s =_\epsilon t$ is derivable (which exists by
(\textbf{arch})). Multiplication collapses terms-over-terms, and the
unit maps an element of $x \in X$ to $[x] \in M_0X$.
\begin{remark}
  The above system for quantitative reasoning follows Ford et
  al.~\cite{DBLP:conf/calco/FordMS21} in slight modifications to the
  original (ungraded) system~\cite{mardare2016quantitative}. In
  particular, we make do without a cut rule, and allow substitution
  only into axioms (substitution into derived equalities is then
  admissible~\cite{DBLP:conf/calco/FordMS21}). We include the rule
  \textbf{(nexp)} ensuring that all operations are nonexpansive,
  i.e.\ the induced graded monad is \emph{enriched} (acts
  nonexpansively on functions).
\end{remark}
\noindent We recall that a graded monad is
\emph{depth-$1$}~\cite{DBLP:conf/calco/MiliusPS15,DBLP:conf/concur/DorschMS19} if $\mu^{nk}$ and
$M_0\mu^{1k}$ are epi-transformations and the diagram below is a
coequalizer of $M_0$-algebras for all~$X$ and $n < \omega$:
\begin{equation}\label{coequalizer}
\begin{tikzcd}
	M_1M_0M_nX \arrow[r,swap, "\mu^{10}M_n", shift right] \arrow[r, "M_1\mu^{0n}", shift left] & M_1M_nX \arrow[r, "\mu^{1n}"] & M_{1+n}X.
\end{tikzcd}
\end{equation}

\noindent By~\autoref{lem:canonical} the following is then immediate:

\begin{proposition}(\cite[Corollary 5.4]{DBLP:conf/concur/DorschMS19})
  \label{canonical}
  If $\mathbb{M}$ is a depth-1 graded monad, then for every $n\in \nat$ and every object~$X$, the $M_1$-algebra with carriers $M_nX$, $M_{n+1}X$ and multiplications as algebra structure is canonical.
\end{proposition}
We briefly refer to canonical algebras as per the above proposition as
being \emph{of the form $M_nX$}.\medskip

\noindent Crucially, we establish a metric variant of a result on depth-$1$
graded monads on $\SET$~\cite{DBLP:conf/calco/MiliusPS15}:

\begin{theorem}
  \label{prop:theoryMonad}
  Graded monads on $\Met$ presented by depth-$1$ graded quantitative
  theories are depth-$1$.
\end{theorem}
\begin{remark}
  A depth-1 graded monad~$\mathbb{M}$ can be reconstructed from its
  constituents of depth at most one, i.e.\ from $M_0$, $M_1$,~$\eta$,
  and the~$\mu^{nk}$ for
  $n+k\le 1$~\cite{DBLP:conf/concur/DorschMS19}.  Graded semantics
  (\autoref{sec:logic}) does however make use of the full structure
  of~$\mathbb{M}$ also at higher depths.
\end{remark}

\phantomsection
\subparagraph*{Presentations of graded trace monads}%
\label{sec:presentations}
We proceed to investigate the quantitative-algebraic presentation of
graded trace monads that are given by a Kleisli distributive law of
the functor
$\Act\times(-)$ (with $\Act$ being the space of action labels) over a
monad (\autoref{expl:graded-monads}). Given a function
$k\colon [0,1]^2\to [0,1]$ with suitable properties, we
write~$\otimes$ for the tensor that equips the Cartesian product of
two sets with the metric
$d_{A\otimes B}((a, b), (a', b')) = k(d(a,a'), d(b, b'))$ generated
by~$k$. This induces trace distances on $\Act^n$, $n\ge 0$, by
viewing~$\Act^n$ as the $n$-fold tensor of~$\Act$. Examples include
the Euclidean ($k(x,y)=\sqrt{x^2+y^2}$), supremum
($k(x,y)=\max(x,y)$), and Manhattan ($k(x,y)=x\oplus y$) distances.
 The fact that $k$ computes distances of traces recursively ``one
symbol at a time'' translates into uniform depth-1 equations:

\begin{defn}\label{def:trace-theory}
  Let $\mathcal{T} = (\Sigma, \mathcal{E})$ be a quantitative
  algebraic presentation of a (plain) monad~$T$ on~$\MET$. We define a
  graded quantitative theory $\mathcal{T}[\Act]$ by including the
  operations and equations of~$\mathcal{T}$ at depth~$0$, along with
  unary depth-1 operations~$a$ for all labels $a \in \Act$, and as
  depth-1 axioms the distributive laws
  $\vdash a(f(x_1, \ldots, x_n)) =_0 f(a(x_1), \ldots, a(x_n))$ for
  all $a\in \Act$ and $f\in \Sigma$, as well as the distance axioms
  $x =_\epsilon y \vdash a(x) =_{k(d(a, b), \epsilon)} b(y)$.
\end{defn}
\noindent The obvious candidate for a Kleisli distributive law
inducing the graded monad presented by the theory $\mathcal{T}[\Act]$
is the family of
maps~$\lambda_X\colon \Act \otimes TX \to T(\Act \otimes X)$ given by
\begin{equation}
\lambda_X(a, t) = T\langle a, \id_X\rangle_{\otimes}(t)\label{eq:lambda}
\end{equation}
where $\langle a, \id_X\rangle_\otimes$ takes $x \in X$ to
$(a, x) \in \Act \otimes X$. However, these maps~$\lambda_X$ may fail
to be nonexpansive, depending on~$T$ and~$\otimes$; for instance, this
happens for $T=\ftFinK$ and~$\otimes$ being Cartesian product~$\times$
(which carries the supremum distance):
\begin{example}
  \label{expl:failure-of-nonexp}
  Put $X=\{x,y\}$ where $d(x, y) = 1$, and
  $s = 0.5\cdot x + 0.5\cdot y$, $t = 1\cdot x\in\ftFinK X$. Clearly
  $d(s, t) = 0.5$. Given $a, b \in \Act$ with $d(a, b) = 0.5$, we have
  $d((a, s), (b, t)) = 0.5$ in $\Act\times\ftFinK X$ while
  $d(\lambda_X(a, s), \lambda_X(b, t)) = d(0.5\cdot(a, x) +
  0.5\cdot(a, y), 1\cdot(b, x)) = 0.75$ in $\ftFinK(\Act\times X)$.
\end{example}

\noindent Nonexpansiveness is, of course, needed to obtain a graded
monad on $\MET$, and as we show later (\autoref{rem:trace-larger}),
its failure may cause undesirable effects. In the case of Manhattan
distance, nonexpansiveness always holds:

\begin{lemma}
  \label{lem:ManhattanNexp}
  The maps $\lambda_X$ as per~\eqref{eq:lambda} are nonexpansive as
  maps $\Act \boxplus TX \to T(\Act \boxplus X)$.
\end{lemma}
\noindent In case the~$\lambda_X$ as per~\eqref{eq:lambda} \emph{are}
nonexpansive, we do in fact have that the distributive law~$\lambda$
and the algebraic theory $\mathcal{T}[\Act]$ induce the same graded
monad:

\begin{lemma}
  \label{lem:distributive}
  Let $\lambda_X$ be defined by~\eqref{eq:lambda}. If
  $\lambda_X\colon \Act \otimes T \to T(\Act \otimes (-))$ is
  nonexpansive for all~$X$, then the~$\lambda_X$ form a Kleisli
  distributive law
  $\lambda\colon\Act \otimes T \to T(\Act \otimes (-))$, and the
  graded monad induced by~$\lambda$ according to
  \autoref{expl:graded-monads} is presented by the quantitative equational
  theory $\mathcal{T}[\Act]$ as per \autoref{def:trace-theory}.
\end{lemma}
\begin{example}\label{expl:trace-theories}
  In our running example of finitely branching metric trace semantics,
  it is easy to check that the distributive law claimed in
  \autoref{expl:graded-monads} is indeed
  nonexpansive, so the induced graded monad is, by
  \autoref{lem:distributive}, presented by the corresponding theory
  as per \autoref{def:trace-theory}, and in particular is
  depth-$1$.
  Explicitly, recall~\cite[Corollary 9.4]{mardare2016quantitative}
  that~$\Pmet$ is a monad, presented in quantitative algebra by the
  usual axioms of join semilattices for a binary join operation~$+$
  and a constant~$0$ %
  (nonexpansiveness of~$+$ is enforced by the deduction rules). The
  quantitative graded theory presenting the graded metric trace monad
  $\Pmet(\Act^n\times{-})$ according to \autoref{lem:distributive} has
  depth-0 operators~$+$ and~$0$ as above and adds unary depth-$1$
  operations~$a$ for all $a \in \Act$, subject to axioms (for
  $a,b \in \Act$, $\epsilon \in [0, 1]$)
  \begin{equation*}
    \vdash a(0) =_0 0\qquad
    \vdash a(x+y) =_0 a(x) + a(y) \qquad x =_\epsilon y \vdash a(x) =_{\epsilon\vee d_{\Act}(a, b)} b(y).
  \end{equation*}
  The distribution of the operations~$a$ over the join semilattice
  structure effectively implements trace equivalence, and the last
  axiom determines the metric on traces, which in this case is taken
  to be the supremum metric.
\end{example}

\section{Graded Quantitative Semantics and Graded
  Logics}\label{sec:logic}

\noindent We proceed to introduce the framework of \emph{graded
  quantitative semantics}, to study spectra of behavioural metrics for
various system types. By `spectra' we informally refer to collections
of process comparisons of varying granularity that arise by observing
a specific system type in different ways, as exemplified by the
classical linear-time/branching-time spectrum on labelled transition
systems~\cite{DBLP:books/el/01/Glabbeek01}. Generally, a \emph{graded
  semantics}~\cite{DBLP:conf/calco/MiliusPS15} $(\mathbb{M},\alpha)$
of a functor $G\colon\catC \to \catC$ consists of a graded
monad~$\mathbb{M}$ and a natural transformation
$\alpha\colon G \to M_1$. Intuitively, $M_n1$ (where~$1$ is a terminal
object of~$\catC$) is a domain of behaviours observable after~$n$
transition steps, with~$\alpha$ determining behaviours after one
step. For a $G$-coalgebra $(X, \gamma)$, we inductively define
\emph{behaviour maps} $\gamma^{(n)}\colon X\to M_n1$ assigning to a
state in~$X$ its behaviour after $n$~steps:
\[\gamma^{(0)}\colon X \xrightarrow{M_0! \cdot \eta} M_01 \hspace{3em}
  \gamma^{(n+1)}\colon X \xrightarrow{\alpha \cdot \gamma} M_1X
  \xrightarrow{M_1\gamma^{(n)}} M_1M_n1 \xrightarrow{\mu^{1n}}
  M_{n+1}1\] For $\catC=\MET$, these maps induce a notion of
\emph{graded behavioural distance} (for readability, we refrain from working
with more general~$\catC$, such as categories of relational
structures~\cite{DBLP:conf/calco/FordMS21}):

\begin{defn}[Graded behavioural distance]
  Given a graded semantics $\alpha\colon G \to M_1$ of a functor~$G$
  on~$\MET$, \emph{(graded) behavioural distance} is the pseudometric on
  states in $G$-coalgebras $(X, \gamma)$ given by
  $d^{\alpha}(x, y) =\bigvee_{n \in \nat}d_{M_n1}(\gamma^{(n)}(x),
  \gamma^{(n)}(y))$ for $x, y \in X$.
\end{defn}

\begin{example}
  \label{expl:graded-semantics}

  The metric trace semantics of finitely branching metric transition
  systems~\cite{afs:linear-branching-metrics,DBLP:journals/tcs/FahrenbergL14}
  and closed-branching metric transition systems is captured by the
  graded metric trace monads $M_n =\Pmet(\Act^n \times {-})$ and
  $M_n =\Pc(\Act^n \times {-})$ (\autoref{expl:graded-monads}),
  respectively (with~$\alpha$ being identity). The behaviour maps
  calculate, at each depth~$n$, sets of length-$n$ traces, whose
  distance is given by the Hausdorff distance induced by the supremum
  metric on traces.

\end{example}

\begin{remark}
\label{rem:trace-larger}
In cases where nonexpansiveness of $\alpha$ or the natural
transformations of $\mathbb{M}$ does not hold (e.g.\ if one attempts
to construct~$\mathbb{M}$ using a family of maps~\eqref{eq:lambda}
that fails to be nonexpansive, cf.\ \autoref{expl:failure-of-nonexp}),
other expected properties can fail. For instance, it can happen that
trace distance exceeds branching time distance (while for trace
semantics induced by nonexpansive graded semantics, general properties
of graded semantics imply that trace distance is below branching-time
distance, in tune with the two-valued setting where trace equivalence
is coarser than bisimilarity\lsnote{Better: Pull branching-time example and lemma
  forward}).  \autoref{expl:failure-of-nonexp} manifests in the
$\ftFinK(\Act \times {-})$-coalgebra (i.e.\ generative probabilistic
metric transition system) shown below, where $\Act = \{a, b, c, d\}$
with relevant distances $d(a,b) = 0.5$ and $d(c,d) = 1$:

\begin{tikzpicture}[node distance=2cm, auto]
  \node at (0, 0) (A)[circle, fill,inner sep=2pt] {};
  \node at ([shift={(160:0.4)}]A) {$x$};
  \node at (1.5, 0) (E)[circle, fill,inner sep=2pt] {};

  \node at (7, 0) (A1)[circle, fill,inner sep=2pt] {};
  \node at ([shift={(45:0.4)}]A1) {$y$};
  \node at (5.5, 0) (D1)[circle, fill,inner sep=2pt] {};

  \node at (4.0, 0) (Sink)[circle, fill,inner sep=2pt] {};

  \path[->,arrows = {-Stealth[length=6pt, inset=2pt]}]
    (A) edge node[below] {\(1\)} node[above] {$a$} (E)

    (Sink) edge[out=220,in=320,looseness=30]  node[below] {\(1\)} node[above] {$a$} (Sink)

    (D1) edge node[below] {$1$} node[above] {$c$}(Sink)

    (E) edge[out=30, in=150] node[above] {$c\quad\frac{1}{2}$} (Sink)
    (E) edge[out=-30, in=-150] node[below] {$d\quad\frac{1}{2}$} (Sink)

    (A1) edge node[below] {\(1\)} node[above=] {$b$}(D1);

\end{tikzpicture}

\noindent Here, we have length-$n$ trace distributions
$\mu^n_x = \frac{1}{2}\cdot(aca^{n-2}) + \frac{1}{2}\cdot(ada^{n-2})$ and
$\mu^n_y = 1\cdot(bca^{n-2})$ for $n\geq 2$.  When the metric on traces is
defined via supremum distance, instead of Manhattan distance as in
\Autoref{sec:prob-met-trace}, the trace distance of the states $x$ and
$y$ is $\bigvee_{n\in \nat}d(\mu^n_x, \mu^n_y) = 0.75$, while their
branching-time distance (cf.\ \autoref{sec:prelims}) is $0.5$.
\end{remark}

\noindent We have the following criterion for invariance of a logic
under a graded semantics $(\alpha,\mathbb{M})$, with $\mathbb{M}$
depth-1, for a functor~$G\colon\MET\to\MET$ \emph{that we fix from now
  on}; recall from \autoref{sec:prelims} that we use $\Omega$ to
denote the unit interval $[0,1]$ equipped with Euclidean distance.

\begin{defn}[Graded logic] \label{def:graded-logic}
  Let $o\colon M_0\Omega \to \Omega$ be an $M_0$-algebra structure on
  $\Omega$. A logic~$\mathcal{L}$ is a \emph{graded logic} (for
  $(\alpha, \mathbb{M})$) if the following hold:
\begin{enumerate}
\item For $n$-ary $p\in \mathcal{O}$, the semantics
  $\sem{p}$ is an $M_0$-algebra homomorphism $(\Omega, o)^n\to(\Omega,
  o)$.
\item\label{item:m1-alg} For each $L \in \Lambda$, there is an
  associated nonexpansive map $\rsem L\colon M_1\Omega\to\Omega$ such
  that the semantics $\sem{L}\colon G\Omega\to\Omega$ factors as
  $\sem{L} =
  (G\Omega\xrightarrow{\alpha_\Omega}M_1\Omega\xrightarrow{\rsem
    L}\Omega)$, and such that the tuple
  $(\Omega, \Omega, o, o, \rsem L)$ constitutes an $M_1$-algebra (that
  is,~$\rsem L$ satisfies homomorphy and coequalization, cf.\
  \autoref{sec:prelims}). We abuse notation and write $\rsem L$ to
  denote the $M_1$-algebra $(\Omega, \Omega,o,o, \rsem L)$.
\end{enumerate}
\end{defn}

\noindent
Notice the different treatment of nullary propositional operators and
truth constants: The former are required to be interpreted as
homomorphisms $1\to(\Omega,o)$ in a graded logic, while no such
condition is imposed on truth constants. In many examples,
$\alpha = id$, in which case condition~\ref{item:m1-alg} just states
that $(\Omega, \Omega, o, o, \sem{L})$ is an $M_1$-algebra
(non-identity~$\alpha$ are associated, for instance, with readiness
and failure semantics~\cite{DBLP:conf/concur/DorschMS19}).

\begin{defn}
  \label{def:inv-expr}
  We say that $\mathcal{L}$ is \emph{invariant} with respect to a
  graded semantics $(\alpha, \mathbb{M})$ if
  $d^\mathcal{L} \leq d^\alpha$ holds in all
  $G$-coalgebras; \emph{expressive} if $d^\mathcal{L} \geq d^\alpha$;
  and \emph{characteristic} if $d^\mathcal{L} = d^\alpha$.
\end{defn}
\begin{theorem}[{\cite[Proposition 21]{DBLP:journals/corr/abs-2307-14826}}]
  \label{thm:invariance}
  Let $\mathcal L$ be a graded logic for $(\alpha, \mathbb M)$. Then
  the evaluation maps~$\sem{\phi}_\gamma$ of uniform-depth $\mathcal{L}$-formulae
  $\phi$ on $G$-coalgebras $(X,\gamma)$ are nonexpansive w.r.t.\
	behavioural distance $d^\alpha$, and hence $\mathcal{L}$ is invariant.
\end{theorem}
\noindent The assumption of uniform depth cannot be removed in
general~\cite{DBLP:journals/corr/abs-2307-14826}.

\begin{example}\label{expl:graded-logics}
  We have a graded logic $\mathcal{L}^\mathsf{mtrace}$ for metric
  trace semantics
  (\autoref{expl:graded-semantics}) %
  featuring modalities~$\modal{a}$ for all $a\in\Act$ as in
  \autoref{expl:modalities}, a single truth constant~$1$, and no
  propositional operators. We equip the set~$\Omega=[0,1]$ of truth
  values with the usual $\Pmet$-algebra (i.e.\ join semilattice)
  structure $([0,1],\vee,0)$, and let $\hat{1}\colon 1 \to [0,1]$ take
  the value~$1$. The logic $\mathcal{L}^\mathsf{mtrace}$ remains
  invariant under metric trace semantics when extended with
  propositional operators that are nonexpansive join-semilattice
  morphisms, such as $\vee$. Analogously we define the logic
  $\mathcal{L}^\mathsf{cmtrace}$ for trace semantics of
  closed-branching metric transition systems. Notice that the
  interpretation of~$1$ fails to be homomorphic, so~$1$ needs to be a
  truth constant.
\end{example}

\section{Expressivity Criteria}\label{sec:expr}

\noindent We proceed to adapt expressivity criteria appearing in
previous work on two-valued behavioural
equivalences~\cite{DBLP:conf/concur/DorschMS19,DBLP:conf/lics/FordMS21}
to the quantitative setting, which poses quite specific challenges. A
key role in the treatment of expressivity of logics will be played by
the notion of initiality~\cite{DBLP:books/daglib/0023249}.

\begin{defn}
  A family of maps $(f_i\colon A\to B)_{i\in I}$ between metric spaces
  $A$ and $B$ is \emph{initial} if~$A$ carries the smallest (pseudo-)metric
  making all maps $f_i$ nonexpansive, explicitly:
  $d(x,y)=\bigvee_i d(f_i(x),f_i(y))$.
\end{defn}

\noindent Using this notion, the definition of expressivity can be
rephrased as follows: An invariant logic~$\mathcal{L}$ is expressive
if for every $G$-coalgebra~$(X,\gamma)$, the family of all evaluation
maps~$\sem{\phi}_\gamma$ of uniform-depth formulae~$\phi$ is initial
on~$(X,d^\alpha)$.

\begin{remark}\label{rem:density}
  In the branching-time case, a stronger notion of expressivity,
  roughly phrased as \emph{density} of the set of depth-$n$ formulae
  in the set of nonexpansive properties at depth~$n$, follows from
  expressivity under certain additional
  conditions~\cite{DBLP:conf/csl/Forster0HNSW23,DBLP:conf/concur/WildS20,DBLP:conf/fossacs/WildS21,DBLP:conf/lics/WildSP018,DBLP:conf/concur/KonigM18}, using
  lattice-theoretic variants of the Stone-Weierstraß theorem. The
  analogue of the Stone-Weierstraß theorem in
  general fails for coarser semantics. Also, for semantics coarser
  than branching time, expressivity in the sense of
  \autoref{def:inv-expr} can often be established using more
  economic sets of propositional operators (e.g.\ no propositional
  operators at all), for which density will clearly fail.
\end{remark}
\noindent Our expressivity result is based on propagating initiality
through an induction on depth. Unlike in the Eilenberg-Moore
case~\cite{DBLP:journals/corr/abs-2307-14826}, this requires, in many
examples, to strengthen the inductive invariant; we treat this
systematically as follows:
\begin{defn}
  An \emph{initiality invariant} is a property~$\Phi$ of sets
  $\mathfrak{A} \subseteq \MET(X, \Omega)$ of nonexpansive functions
  such that~(i) every family of maps
  satisfying~$\Phi$ is initial, and~(ii)~$\Phi$ is upwards closed
  w.r.t.\ subset inclusion.
\end{defn}
\begin{example}\label{expl:invariants}
  \begin{enumerate}[wide]
  \item Initiality itself is an initiality invariant. If~$\Phi$ is
    initiality, then we say `initial-type' for `$\Phi$-type'. %
  \item We say that $\mathfrak{A}\subseteq\MET(X,\Omega)$ is
    \emph{normed isometric} if whenever $d(x,y)>\epsilon$ for
    $x,y \in X$ and $\epsilon>0$, then there is some
    $f \in \mathfrak A$ such that $|f(x)-f(y)|> \epsilon$ and
    $f(x)\vee f(y)=1$. Normed isometry is an initiality invariant. %
  \end{enumerate}
\end{example}
\noindent Our expressivity criterion then takes the following shape:
\begin{defn}
  Let~$\Phi$ be an initiality invariant. A graded
  logic~$\mathcal{L}=(\Theta, \mathcal{O}, \Lambda)$ with truth value
  object $(\Omega,o)$ is \emph{$\Phi$-type depth-0 separating}
  if %
  the family of maps
  $\{o \cdot M_0\hat{c} \colon M_0 1 \to \Omega \mid c \in
  \Theta\}$ %
  has property~$\Phi$. Moreover,~$\mathcal{L}$
  is \emph{$\Phi$-type depth-1 separating} if whenever~$A$ is a
  canonical $M_1$-algebra of the form $M_n1$
  (\autoref{canonical}) and~$\mathfrak{A}$ is a set of
  $M_0$-homomorphisms $A_0 \rightarrow \Omega$ that has
  property~$\Phi$ and is closed under the propositional operators
  in~$\mathcal{O}$, then %
  the set
	\[\Lambda(\mathfrak{A}) := \{\appModal{L}{g}: A_1 \rightarrow \Omega \mid L \in
	\Lambda, g \in \mathfrak{A} \}\] %
      has property $\Phi$, where
      $\appModal{L}{g}\colon A_1 \to \Omega$ is the (by canonicity,
      unique) morphism extending the $M_0$-algebra morphism $g$ to an
      $M_1$-algebra morphism $A\to\rsem L$
      (\autoref{def:graded-logic}).
\end{defn}
\begin{theorem}[Expressivity]
  \label{thm:main}
  Let~$\Phi$ be an initiality invariant, and suppose that a graded
  logic $\mathcal{L}$ is both $\Phi$-type depth-0 separating and
  $\Phi$-type depth-1 separating. Then $\mathcal{L}$ is expressive.
\end{theorem}
\begin{remark}\label{rem:restrict-separation}
  Our definition of separation differs from notions used for
  two-valued
  logics~\cite{DBLP:conf/concur/DorschMS19,DBLP:conf/lics/FordMS21}
  and for quantitative graded semantics induced by Eilenberg-Moore
  distributive laws~\cite{DBLP:journals/corr/abs-2307-14826}, which
  overall have turned out to be much more well-behaved than the more
  general setting of the present work. %
  The most obvious novelty is the use of an initiality
  invariant~$\Phi$ strengthening the induction hypothesis in the
  inductive proof of \autoref{thm:main}. We will see that this is
  needed even in very simple examples in our more general setting.
  Moreover, we have phrased separation in terms of the specific
  canonical algebras $M_n1$ on which it is needed, rather than on
  unrestricted canonical algebras. This allows exploiting additional
  properties of $M_n1$, e.g.\ that for graded monads $M_n=TF^n$
  arising from Kleisli distributive laws
  (\autoref{expl:graded-monads}), $M_n1$ is free as an $M_0$-algebra.

\end{remark}

\begin{example}
  \begin{enumerate}[wide]
    \label{expl:separation}
  \item \emph{Metric Streams:} A simple example for failure of
    initial-type separation (\autoref{expl:invariants}) are metric
    streams, i.e.\ streams over a metric space of labels
    $(\Act, d_\Act)$; these are coalgebras for the functor
    $G = \Act \times {-}$.  Behavioural distance on streams is
    captured by the graded monad $G^n = \Act^n \times
    \{-\}$. %
    The logic~$\mathcal{L}$ consisting of the truth constant~$1$
    and modalities~$\modal{a}$ for all $a\in\Act$, with
    interpretation
    $\sem{\modal{a}}\colon \Act\times [0,1] \to [0,1]$ given by
    $(b, v) \mapsto (1- d_\Act(a, b)) \wedge v$,\lsnote{@Lutz: Add intuition} %
    is $\Phi$-type depth-0 separating and $\Phi$-type depth-1
    separating for~$\Phi$ being normed isometry, and hence expressive
    by \autoref{thm:main}. (The modality~$\Diamond_a$ restricts the
    corresponding modality for metric transition systems as per
    \autorefexpls{expl:modalities} and~\ref{expl:graded-logics} to
    metric streams: a state satisfies $\Diamond_a\phi$ to the
    degree that its output is close to~$a$ and its successor
    satisfies~$\phi$). On the other hand,~$\mathcal{L}$ fails to be
    initial-type depth-1 separating, illustrating the necessity of the
    general form of \autoref{thm:main}.
    
    \item \label{expl:mts-separation}\emph{Metric transition systems:} The
  graded logics $\mathcal{L}^\mathsf{mtrace}$ and
  $\mathcal{L}^\mathsf{cmtrace}$ for metric trace semantics
  (\autoref{expl:graded-logics}), in the version with no propositional
  operators, are $\Phi$-type depth-$0$ separating and $\Phi$-type
  depth-$1$ separating for $\Phi$ being normed isometry, and hence are
  expressive by \autoref{thm:main}. We thus improve on an example from
  recent work based on Galois
  connections~\cite{DBLP:conf/csl/BeoharG0M23}, where application of
  the general framework required the inclusion of propositional shift
  operators (which were subsequently eliminated in an ad-hoc manner),
  and we generalize to systems with closed branching on a metric state
  space.  %

\item Probabilistic metric trace semantics is modelled
  straightforwardly as a graded semantics using a graded trace
  monad~(\autoref{expl:graded-monads}). By~\autoref{thm:propNoExpr},
  however, there is no graded logic for probabilistic metric trace
  semantics that satisfies the conditions of \autoref{thm:main}.
\end{enumerate}
\end{example}

\begin{remark}\label{rem:galois}
  In a recent approach based on Galois
  connections~\cite{DBLP:conf/csl/BeoharG0M23,DBLP:conf/stacs/BeoharG0MFSW24},
  logics are related to fixpoints of behaviour functions induced by
  the logic itself (similar to approaches that define trace semantics
  via intended characteristic
  logics~\cite{DBLP:journals/corr/KlinR16}), while our present
  interest is in providing logical characterizations of \emph{given}
  behavioural distances. %
  The Galois framework is highly general, and in fact not even tied to
  coalgebraic modelling, or in fact to state-based systems of any
  kind~\cite{DBLP:conf/csl/BeoharG0M23}, but correspondingly offers
  less concrete recipes. Instantiated to our current setup, the key
  condition of \emph{compatibility} appearing in \emph{op.\ cit.\
  }%
  roughly speaking amounts to initial-type depth-1 separation of the
  logic w.r.t.\ its own Kantorovich
  lifting~\cite{bbkk:trace-metrics-functor-lifting,DBLP:conf/stacs/BeoharG0MFSW24}.
\end{remark}

\begin{remark}[Branching-time semantics]\label{expl:branching-time}
  Any functor $G$ yields a graded monad given by iterated application
  of $G$, that is $M_n = G^n$, and by unit and multiplication being
  identity~\cite{DBLP:conf/calco/MiliusPS15}.  In general, the
  finite-depth branching-time semantics of a $G$-coalgebra
  $(X, \gamma)$ is defined via its \emph{canonical cone}
  $(p_i\colon X \to G^i1)_{i<\omega}$ into the \emph{final sequence}
  $1 \xleftarrow{!} G1 \xleftarrow{G!} G^21 \leftarrow\dots$
  of~$G$. The $p_i$ are defined inductively by
  $p_0 = \mathord{!} \,\colon X\to 1$ and
  $p_{i+1} = Gp_i \cdot \gamma$. This semantics is captured by the
  graded monad $M_n = G^n$ and
  $\alpha=\id$~\cite{DBLP:conf/calco/MiliusPS15}.  More specifically,
  the \emph{finite-depth branching-time behavioural distance} of
  states $x,y\in X$ is $\bigvee_{i<\omega} d(p_i(x),p_i(y))$, and thus
  agrees with the graded behavioural distance obtained via the graded
  semantics in the graded monad $M_n = G^n$.  This monad has
  $M_0=\id$, so that the corresponding graded logics are just
  branching-time logics without further
  restriction~\cite{DBLP:conf/calco/MiliusPS15,DBLP:conf/concur/DorschMS19}. Coalgebraic
  quantitative logics of this kind have received some recent
  attention~\cite{DBLP:conf/csl/Forster0HNSW23,DBLP:conf/concur/WildS20,DBLP:conf/fossacs/WildS21,DBLP:conf/lics/WildSP018,afs:linear-branching-metrics,JainEA20,DBLP:conf/concur/KonigM18}.
  Suppose~$\Lambda$ is a finite \emph{separating} set of modalities,
  i.e.\ the maps $\sem{L} \cdot Gf\colon GX\to\Omega$, with~$L$
  ranging over modalities and~$f$ over nonexpansive maps
  $X\to\Omega$, form an initial family. Moreover, let~$\mathcal{O}$
  contain truth~$1$, meet~$\land$, %
  fuzzy negation $\neg$ (i.e.\ $\neg x=1-x$), and truncated addition
  of constants $(-)\oplus c$. Then one shows using a variant of the
  Stone-Weierstraß theorem~\cite{DBLP:conf/lics/WildSP018} that the
  graded logic~$\mathcal{L}$ given by $\Lambda$, $\mathcal{O}$, and
  $\Theta=\emptyset$ is initial-type depth-$0$ separating and
  initial-type depth-$1$ separating. By \autoref{thm:main}, we obtain
  that~$\mathcal{L}$ is expressive. Previous work on quantitative
  branching-time
  logics~\cite{DBLP:conf/lics/WildSP018,DBLP:conf/concur/KonigM18,DBLP:conf/concur/WildS20,DBLP:conf/fossacs/WildS21,DBLP:conf/csl/Forster0HNSW23}
  discusses, amongst other things, conditions on~$G$ that allow
  concluding expressivity even for infinite-depth behavioural
  distance.
\end{remark}

\section{Case Study: Fuzzy Metric Trace Semantics}
\label{sec:examples}
\lsnote{@Lutz: Add intuition}
We apply the recipe outlined above to
obtain a characteristic logic for trace distance on \emph{fuzzy
  metric transition systems}. That is, we proceed as follows: We cast
fuzzy metric trace distance as a graded semantics using a suitable
depth-1 graded monad~$\mathbb{M}$, and check that~$\mathbb{M}$ is
depth-1 using the techniques outlined in
\autoref{sec:presentations}. We then identify a corresponding graded
logic~$\mathcal{L}$, verifying the requirements of
\autoref{def:graded-logic}. Invariance of $\mathcal{L}$ then follows
automatically (\autoref{thm:invariance}). Finally, we show
expressivity using \autoref{thm:main}.

A \emph{fuzzy $\Act$-labelled metric transition system (fuzzy metric
  LTS)}~\cite{DErricoLoreti07,WuEA18,WuEA18b,JainEA20}) consists of a
set (or metric space)~$X$ of states and a fuzzy transition relation
$R\colon X\times\Act\times X\to[0,1]$, with~$\Act$ a metric space. A
fuzzy LTS $(X,R)$ is \emph{finitely branching} if
$\{(a,y)\mid R(x,a,y)>0\}$ is finite for every~$x\in X$. Equivalently,
a finitely branching fuzzy LTS is a coalgebra for the functor
$\FHaus_\omega(\Act\times (-))$ (cf.\
\autoref{expl:monads}.\ref{item:fuzzy-power}).
	
  A natural fuzzy trace semantics of fuzzy transition systems assigns
  to each state~$x$ of a fuzzy LTS $(X,R)$ a fuzzy trace set
  $\Trace(x)\in\FPow_\omega (\Act^*)$ where
    \begin{equation*}\textstyle
      \Trace(x)(a_1\dots a_n)=\bigvee\{\bigwedge_{i=1}^n
      R(x_{i-1},a_i,x_i)\mid x=x_0,x_1,\dots,x_n\in X\}.
    \end{equation*}
    This notion of trace relates, for instance, to a notion of fuzzy
    path that is implicit in the semantics of fuzzy computation tree
    logic~\cite{PanEA16} and to notions of fuzzy language accepted by
    fuzzy automata~(e.g.~\cite{Belohlavek02}). We obtain a notion of
    \emph{fuzzy trace distance}~$d^T$ of states~$x,y$, given by the
    distance of $\Trace(x)$, $\Trace(y)$ in $\FHaus_\omega(\Act^*)$,
    i.e.\ under fuzzy Hausdorff distance
    (\autoref{expl:monads}.\ref{item:fuzzy-power}) w.r.t.\ the metric
    on~$\Act^*$ that is the supremum metric on each~$\Act^n$, and
    assigns distance~$1$ to traces of different lengths. To capture
    this distance in a graded semantics, consider the distributive law
    $\lambda\colon \Act \times \FHaus_\omega({-}) \to
    \FHaus_\omega(\Act \times {-})$ given by
    $\lambda(a, U)(a,x)= U(x)$ and $\lambda(a, U)(b, x) = 0$ for
    $b \neq a$. By \autoref{expl:graded-monads} we thus obtain the
    \emph{graded fuzzy metric trace monad}
    $M_n =\FHaus_\omega(\Act^n \times (-))$.
    \label{item:fuzzy-theory}
    The monad $\FHaus_\omega$ can be presented by the following
    quantitative equational theory: Take a
    binary operation~$+$, a constant~$0$, and unary operations $r$ for
    every $r\in[0,1]$. Impose strict equations ($=_0$) saying
    that~$+$,~$0$ form a join semilattice structure and that the
    operations~$r$ define an action of the monoid $([0,1],\wedge)$
    (i.e.\ $1(x)=x$, $r(s(x))=_0(r\wedge s)(x)$). Finally, impose
    axioms $x=_\epsilon y\vdash r(x)=_\epsilon s(y)$ for $r,s\in[0,1]$
    such that $|r-s|\le\epsilon$. By \autoref{lem:distributive}, the
    graded fuzzy trace monad $M_n = \FHaus_\omega(\Act^n\times X)$ is
    presented by the above algebraic description of~$\FHaus_\omega$ at
    depth~$0$, with additional depth-1 unary operations~$a$ for
    $a\in\Act$ and depth-1 equations $a(x+y)=_0 a(x)+a(y)$,
    $a(0)=_0 0$, $a(r(x))=_0 r(a(x))$, and
    $x =_\epsilon y \vdash a(x)=_{\epsilon \vee d(a,b)} b(y)$.

    \textbf{Fuzzy metric trace logic} \label{item:fuzzy-trace-logic}
    interprets the additional operations $r\in[0,1]$ on the truth
    value object $[0,1]$ by $r(x)=r\wedge x$, and otherwise uses the
    same quantitative join semilattice structure as for metric trace
    semantics (\autoref{expl:graded-logics}). We include the truth
    constant~$1$ and modal operators~$\modal{a}^c$ for $a\in \Act$ and
    $c \in [0,1]\cap\mathbb{Q}$, with interpretation
    $\sem{\modal{a}^c}\colon M_1[0,1]\to[0,1]$ given by
    $\sem{\modal{a}^c}(A)=\bigvee_{b\in \Act, v\in [0,1]}A(b,v)\wedge
    v \wedge (c\ominus d(a,b))$. (When~$\Act$ is discrete,
    then~$\modal{a}^1$ is the usual fuzzy diamond modality,
    e.g.~\cite{Fitting91}). Thus, a state~$x$ in a fuzzy metric
    transition system satisfies $\Diamond^c_a\phi$ to the degree
    that~$x$ has a $b$-successor~$y$ with~$b$ close to~$a$ and~$y$
    satisfying~$\phi$; crucially, `closeness' of~$b$ to~$a$ needs to
    be shifted down as governed by the parameter~$c$. This logic is
    initial-type depth-$0$ separating and initial-type depth-$1$
    separating, and hence expressive for fuzzy trace distance by
    \autoref{thm:main}; both this result and the logic itself appear
    to be new (the case with~$\Act$ discrete is partially covered in
    work on Galois connections~\cite{DBLP:conf/stacs/BeoharG0MFSW24}).
    Indeed for non-discrete~$\Act$, the logic with only~$\modal{a}^1$
    instead of all $\modal{a}^c$ fails to be expressive. The logic
    remains invariant when extended with additional nonexpansive
    propositional operators that are $\FHaus_\omega$-homomorphic, such
    as~$\vee$.

\section{Conclusions}\lsnote{Future work on real-valued measures instead of probabilities?}

\noindent We have shown that there is no unary quantitative
coalgebraic modal logic characterizing a natural notion of
quantitative trace distance on probabilistic metric transition
systems. Moving onwards from this observation, we have developed a
generic framework for linear-time/branching-time spectra of
behavioural distances on state-based systems in coalgebraic
generality, covering, for instance, metric, probabilistic, and fuzzy
transition systems. Unlike previous work on Eilenberg-Moore-style
coalgebraic trace
distances~\cite{DBLP:conf/stacs/BeoharG0MFSW24,DBLP:journals/corr/abs-2307-14826},
the framework covers also systems with labels from a metric space. The
key abstractions in the framework are based on the notion of a graded
monad on the category of metric spaces and an arising notion of
quantitative graded semantics. We have provided a graded quantitative
algebraic system for the description of such graded monads (extending
and modifying the existing non-graded
system~\cite{mardare2016quantitative}). %
Moreover, we have established sufficient conditions for canonical
invariant \emph{quantitative graded
  logics}~\cite{DBLP:journals/corr/abs-2307-14826} to be
\emph{expressive} for given quantitative graded semantics, and we have
exploited this result to obtain expressive logics for some instances
of Kleisli-type trace semantics~\cite{DBLP:journals/lmcs/HasuoJS07},
notably including a new result for fuzzy metric trace
semantics. %

One important next step in the development will be to identify a
generic game-based characterization of behavioural distances in the
framework of graded semantics, generalizing work specific to metric
transition systems~\cite{DBLP:journals/tcs/FahrenbergL14} and building
on game-based concepts for two-valued graded
semantics~\cite{DBLP:conf/lics/FordMSB022}. Also, there is interest in
computing distinguishing quantitative formulae
(cf.~\cite{Cleaveland90,WissmannEA21} for the two-valued
branching-time setting), generalizing recent results for the
branching-time case~\cite{RadyBreugel23} to spectra of coarser
semantics.

 \bibliographystyle{plainurl}
 \bibliography{gsq}

 \newpage \appendix

 \section{Appendix: Additional Details and Omitted Proofs}

 We give details and proofs omitted in the main body.

\subsection{Details for \autoref{sec:prelims}
(\nameref{sec:prelims})}

\subsubsection*{Details for \autoref{expl:monads}}

We show explicitly that $\Pc \colon \Met \to \Met$ is a monad.
For the unit laws, we need to show that
\begin{enumerate}
    \item $\mu_X \circ \Pc\eta_X = \id_{\Pc X}$: Let $Y \in \Pc X$, by
      unraveling of definitions
     we have to show that $\mu \circ \Pc\eta_X(Y)$ is the closure of $Y$, which is already closed.
    \item $\mu_X \circ \eta_{\Pc X} = \id_{\Pc X}$: Let $Y \in \Pc X$,
      then we again have to show that  $\mu \circ \eta_{\Pc X}(Y)$ is the closure of $Y$, which is already closed.
\end{enumerate}
\noindent For assosciativity, let $Y \in \Pc \Pc \Pc X$. We show that in both paths of the diagram, the result is the closure of $\{w \mid w \in z, z \in y, y\in Y\}$. On the one hand we have 

\begin{equation*}
    \begin{split}
        \mu_{X} \cdot \Pc\mu_{X}(Y) &= \mu_{X}(\{ \mu_{X}(Z) \mid Z \in Y \})\\
        &=  \mu_{X}(\{ \overline{ \bigcup_{W\in Z}  W  } \mid Z \in Y\} )\\
        &= \overline{\bigcup_{Z\in Y} \overline{\bigcup_{W\in Z}  W }}\\
        &= \overline{\bigcup_{Z\in Y} \bigcup_{W\in Z}  W }\\
    \end{split}
\end{equation*}

\noindent For the seccond path we utilize the following lemma

\begin{lemma}
    \label{lem:closure}
    Let $X$ be a metric space,  and let $A\subseteq \mathcal{P}(X)$ be a subspace of $\mathcal{P}(X)$ equipped with the Hausdorff metric. Then $\overline{\bigcup_{Y \in\overline{A}} Y} = \overline{\bigcup_{Y \in A} Y}$.
\end{lemma}
\begin{proof}
    The inclusion $\supseteq$ is obvious. For the inclusion $\subseteq$
    pick $x\in \overline{\bigcup_{Y \in\overline{A}} Y}$. Let $\epsilon > 0$. We need to show that there is $y \in \bigcup_{Y \in A} Y$ such that $d(x, y) < \epsilon$. By definition of closure, there is $y' \in \bigcup_{Y \in\overline{A}} Y$ such that $d(x, y') < \frac{\epsilon}{2}$. This implies there is $Y \in \overline{A}$ with $y' \in Y$. Further, by definition of closure there must be $Y' \in A$ such that $d(Y, Y') < \frac{\epsilon}{2}$. By definition of the Hausdorff metric, there must be $y \in Y'$ such that $d(y, y') < \frac{\epsilon}{2}$. So $d(x, y) \leq d(x, y') + d(y', y) \leq \epsilon$.
\end{proof}

\noindent Then it follows that

\begin{equation*}
    \begin{split}
        \mu_{X} \cdot \mu_{\Pc X}(Y) &= \mu_X(\overline{\bigcup_{Z\in Y} Z})\\
        &= \overline{\bigcup_{W \in\overline{\bigcup_{Z\in Y} Z}} W}\\
        &= \overline{\bigcup_{W \in\bigcup_{Z\in Y} Z} W} \by{\autoref{lem:closure}}\\
        &= \overline{\bigcup_{Z\in Y} \bigcup_{W\in Z}  W }\by{Associativity of monad multiplication of $\mathcal{P}$}\\
    \end{split}
\end{equation*}

\subsection{Details for \autoref{sec:prob-met-trace}
(\nameref{sec:prob-met-trace})}
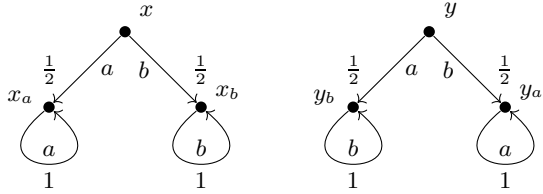
\begin{figure}[h]
	\label{system}
	\caption{States with probabilistic trace distance $v$ and logical distance $< v^2$.}
\centering
\begin{tikzpicture}[node distance=2cm, auto]
				\node at (2.5, 1) (A)[circle, fill,inner sep=2pt] {};
				\node at ([shift={(45:0.4)}]A) {$x$};
				\node at (1.5, 0) (D)[circle, fill,inner sep=2pt] {};
				\node at ([shift={(160:0.4)}]D) {$x_a$};
				\node at (3.5, 0) (E)[circle, fill,inner sep=2pt] {};
				\node at ([shift={(30:0.4)}]E) {$x_b$};

				\node at (6.5, 1) (A1)[circle, fill,inner sep=2pt] {};
				\node at ([shift={(45:0.4)}]A1) {$y$};
				\node at (5.5, 0) (D1)[circle, fill,inner sep=2pt] {};
				\node at ([shift={(160:0.4)}]D1) {$y_b$};
				\node at (7.5, 0) (E1)[circle, fill,inner sep=2pt] {};
				\node at ([shift={(30:0.4)}]E1) {$y_a$};

				\path[->,arrows = {-Stealth[length=6pt, inset=2pt]}]
					(A) edge node[left=2ex] {\(\frac{1}{2}\)} node[right=0.5ex] {$a$} (D)
						edge node[right=2ex] {\(\frac{1}{2}\)} node[left=0.5ex] {$b$} (E)

					(D) edge[out=220,in=320,looseness=30]  node[below] {\(1\)} node[above] {$a$} (D)
					(E) edge[out=220,in=320,looseness=30]  node[below] {\(1\)} node[above] {$b$} (E)

					(D1) edge[out=220,in=320,looseness=30]  node[below] {\(1\)} node[above] {$b$} (D1)
					(E1) edge[out=220,in=320,looseness=30]  node[below] {\(1\)} node[above] {$a$} (E1)

					(A1) edge node[left=2ex] {\(\frac{1}{2}\)} node[right=0.5ex] {$a$}(D1)
						 edge node[right=2ex] {\(\frac{1}{2}\)} node[left=0.5ex] {$b$}(E1) ;

\end{tikzpicture}
\label{fig:propCtrEx}
\end{figure}

\subsubsection*{Proof of \autoref{thm:propNoExpr}}
For the sake of deriving a contradiction, assume that $\mathcal{L}$ is
invariant and expressive.  Consider the transition systems pictured in
\autoref{fig:propCtrEx}. Here, $a, b \in \Act$ with $a \not = b$ and
$d(a,b) = v < 1$. The depth-$2$ behaviour of~$x$ is
$0.5(a, a) + 0.5(b, b)$, and that of~$y$ is $0.5(a, b) + 0.5(b, a)$.
It is easy to see that the behavioural distance of these length-$2$
trace distributions is $v$. We show that the deviations of truth
values of $\mathcal{L}$-formulae between~$x$ and~$y$ stay away from
that value, specifically below~$v^2<v$, contradicting expressivity.

First, note that~$x$ and~$y$ are behaviourally equivalent up to
depth-$1$, and hence agree on formulae of depth at most~$1$; we thus
restrict attention to formulae of depth at least~$2$ (where we do not
insist on uniform depth but exclude top-level depth-$1$ subformulae).
Let~$L\in\Lambda$ be a modal operator, interpreted as
$\sem{L}\colon \ftFinK(\Act \boxplus [0, 1]) \to [0,1]$.  Then we have
the following:

\begin{claim}\label{claim:dist}
  The interpretation~$\sem{L}$ satisfies the equation
\begin{equation}
	\label{eq:distribution}
	\sem{L}\Big( \sum_{i \in I}p_i (a_i, \sum_{j\in J_i} q_{ij} v_{ij} )\Big) = \sem{L}\Big( \sum_{i \in I, j\in J_i} p_i q_{ij} (a_i,   v_{ij} )\Big)
\end{equation}
where the outer sum on the left hand side and the sum on the right
hand side are formal sums describing probability distributions on
$\Act\boxplus[0,1]$, while the inner sum on the left hand side is just an
arithmetic sum.%
\end{claim}

\begin{proof}[Proof (\autoref{claim:dist})]
    Let $x_0$ and $x_1$ be states in coalgebras with probabilistic trace distance~$1$ (such states exist because $\Act$ contains labels with positive distance and the distance on $\Act^{\boxplus n}$ is Manhattan distance). By expressivity,  for every $\epsilon > 0$, there is a formula $\phi_\epsilon$ such that $d(\sem{\phi_\epsilon}_\gamma(x_0), \sem{\phi_\epsilon}_\gamma(x_1)) > 1 - \epsilon$ (w.l.o.g. we assume $\sem{\phi_\epsilon}(x_0) < \epsilon$ and $\sem{\phi_\epsilon}(x_1)> 1-\epsilon$). Now we construct a new coalgebra $(X', \gamma')$ extending~$(X,\gamma)$ with new states $x_p$ for $p \in [0,1]$, with successor distributions $\gamma'(x_p) = \gamma(x_1) +_p \gamma(x_0)$. The probabilistic trace distance of $x_0$ and $x_p$ is $p$, while the probabilistic trace distance between $x_p$ and $x_1$ is $1-p$. By invariance, $p - \epsilon < \sem{\phi_\epsilon}(x_p) < p + \epsilon$. We construct a coalgebra $(X'',\gamma'')$ that extends
    $(X',\gamma')$ with two new states $l, r$, which receive
  successor distributions
  \begin{equation*}
    \gamma''(l) = \sum_{i\in I}p_i(a_i, \sum_{j\in J_i} q_{ij} x_{v_{ij}}) \quad
    \gamma''(r) = \sum_{i\in I, j\in J_i} p_iq_{ij}(a_i, x_{v_{ij}})
  \end{equation*}
  It is clear that $l$ and $r$ are probabilistically trace
    equivalent. Now $\ftFinK(\Act \boxplus \sem{\phi_\epsilon}) \circ \gamma''(l)$ has distance  at most $\epsilon$ from the argument on the left hand side of \eqref{eq:distribution}, and symmetrically $\ftFinK(\Act \boxplus \sem{\phi_\epsilon}) \circ \gamma''(r)$ has distance  at most $\epsilon$ from the argument on the right hand side. Since~$\sem{L}$ is nonexpansive (being the interpretation of a modal operator in a coalgebraic modal logic over $\MET$), we have that
	\begin{equation}
		\begin{split}
			&\sem{L}\left( \sum_{i \in I}p_i (a_i, \sum_{j\in J_i} q_{ij} v_{ij} )\right) = \\
			&\sem{L} \left(\lim_{\epsilon \to 0} \ftFinK(\Act \boxplus \sem{\phi_\epsilon})\left(\sum_{i\in I}p_i(a_i, \sum_{j\in J_i} q_{ij} x_{v_{ij}}) \right) \right) =\\
			&\sem{L} \left( \lim_{\epsilon \to 0} \ftFinK(\Act \boxplus \sem{\phi_\epsilon})\left(\sum_{i\in I, j\in J_i} p_iq_{ij}(a_i, x_{v_{ij}}) \right) \right) = \\
			&\sem{L}\left( \sum_{i \in I, j\in J_i} p_i q_{ij} (a_i, v_{ij} )\right)\\
		\end{split}
	\end{equation}
\end{proof}

\noindent %
We call distributions from $\ftFinK(\Act \boxplus [0,1])$
\emph{$L$-equivalent} if they are identified under~$\sem{L}$. Then any
distribution $\sum_I p_i (a_i, v_i)$ is, by
equation~\eqref{eq:distribution}, $L$-equivalent to one of the form
$\sum_I p_iv_i (a_i, 1) + p_i(1-v_i)(a_i, 0)$. Now let~$\phi$ be a
depth-1 formula, and put
$v_a=\sem{\phi}_\gamma(x_a) =
\sem{\phi}_\gamma(y_a)$ (where the second equality holds
because~$x_a$ and~$y_a$ are probabilistically trace equivalent) and
$v_b=\sem{\phi}_\gamma(x_b) =
\sem{\phi}_\gamma(y_b)$. Then $|v_a-v_b| \leq v$ by
invariance. We have
\begin{align*}
  \sem{L\phi}_\gamma(x) &= \sem{L} \cdot \ftFinK(\Act \boxplus \sem{\phi}_\gamma)(0.5(a, x_a) + 0.5(b, x_b))= \sem L(0.5(a, v_a) + 0.5(b, v_b))\\
  \sem{L\phi}_\gamma(y) &= \sem{L} \cdot \ftFinK(\Act \boxplus \sem{\phi}_\gamma)(0.5(a, y_b) + 0.5(b, y_a)) = \sem{L}(0.5(a, v_b) + 0.5(b, v_a)).
\end{align*}
The distribution $\sem{L\phi}_\gamma(x)$ is thus $L$-equivalent to
$\mu = 0.5v_a(a, 1) + 0.5(1-v_a)(a, 0) + 0.5v_b(b, 1) + 0.5(1-v_b)(b,
0)$, while $\sem{L\phi}_\gamma(y)$ is $L$-equivalent to
$\nu = 0.5v_b(a, 1) + 0.5(1-v_b)(a, 0) + 0.5v_a(b, 1) + 0.5(1-v_a)(b,
0)$\lsnote{Instead need to work with $\epsilon$, $1-\epsilon$ here;
  what if, say, $v_a=0$, though?}. We give an upper bound on
$d(\mu, \nu)$.  Without loss of generality, assume $v_a \leq
v_b$. Then the following distribution $\pi$ is a coupling of $\mu$ and
$\nu$:
\begin{equation}
\begin{alignedat}{3}
	\label{eq:coupling}
	\pi &= 0.5v_a((a,1),(a,1)) &&+ 0.5(1-v_b)((a,0),(a,0))\\
	&+ 0.5v_a((b,1),(b,1)) &&+ 0.5(1-v_b)((b,0),(b,0))\\
	&+ 0.5(v_b-v_a)((b,1), (a,1)) &&+ 0.5(v_b-v_a)((a,0), (b,0))
\end{alignedat}
\end{equation}
By the Kantorovich-Rubinstein duality, $d(\mu, \nu)$ is bounded by the expected value
by~$\mathbb{E}_\pi d$ of $d$ under $\pi$. In the calculation of $\mathbb{E}_\pi d$,
the first two lines in \Cref{eq:coupling} contribute~$0$. Thus,
\[
	\begin{split}
          \mathbb{E}_\pi d &= 0.5(v_b-v_a)d((b,1), (a,1)) + 0.5(v_b-v_a)d((a,0), (b,0))\\
		   &\leq 0.5vd((b,1), (a,1)) + 0.5vd((a,0), (b,0))\\
		   &= 0.5v^2 + 0.5v^2\\
		   &= v^2\\
	\end{split}
\]
where in the inequality, we use that $v_b-v_a\le v$.  Since
$\sem{L}$ is nonexpansive, $d(\mu, \nu) \leq v^2$ implies
$|\sem{L}(\mu)-\sem{L}(\nu)|\leq v^2$, and therefore
$|L\phi(x)-L\phi(y)|\le v^2$, as claimed. \qed

\begin{remark}\label{rem:set-based-logics} We note that for coalgebraic logics in the sense we employ
  here, which live natively in the category of metric spaces,
  interpretations~$\sem{L}\colon\ftFinK(\Act\boxplus[0,1])\to[0,1]$ of
  modalities~$L$ are nonexpansive by definition. We can strengthen our
  impossibility result by letting the modal logic live over~$\SET$, in
  the sense that interpretations~$\sem{L}$ can be unrestricted maps;
  for purposes of the present discussion, we refer to such modalities
  as \emph{set-based}. Under a mild additional assumption on the
  logic, we can show that this relaxation does not actually provide
  additional room for maneuvering:

  \begin{claim}\label{claim:lambda-nexp}
    If $\mathcal{L}$ is an invariant and expressive set-based
    coalgebraic logic for probabilistic metric trace semantics, and
    there is a formula $\phi$ and states $x_0$ and $x_1$ in a
    coalgebra $(X, \gamma)$ such that $\sem{\phi}(x_0) = 0$ and
    $\sem{\phi}(x_1) = 1$, then for every modal operator~$L$
    in~$\mathcal{L}$, the interpretation~$\sem{L}$ is nonexpansive.
\end{claim}
\begin{proof}
  Assume for the sake of contradiction that there are
  $s, t \in \ftFinK(\Act \boxplus [0,1])$, such that
  $d(\sem{L}(s), \sem{L}(t)) > d(s,t)$. Since the logic is assumed to
  be invariant, $x_0$ and $x_1$ must have probabilistic metric trace
  distance~$1$. %
  We define a new coalgebra $(X',\gamma')$ extending $(X,\gamma)$ with
  new states $\bar{x}, \bar{y}$, and~$x_p$ for $p \in (0,1)$, with
  successor distributions given by
  $\gamma'(x_p) = \gamma(x_1) +_p \gamma(x_0)$,
  $\gamma'(\bar{x}) = s'$, and $\gamma'(\bar{y}) = t'$, where
  $s',t' \in \ftFinK(\Act \boxplus X')$ are defined by substituting
  all values $p \in [0,1]$ in $s,t$ with $x_p$. The state $x_p$ is
  easily seen to have trace distance $p$ (respectively $1-p$) from
  $x_0$ ($x_1$). Therefore, because of invariance,
  $\sem{\phi}(x_p) = p$.  The probabilistic metric trace distance
  of~$\bar x$ and~$\bar y$ is at most $d(s', t')=d(s,t)$. On the other
  hand,
  \begin{equation*}
    \begin{split}
      &d(\sem{L\phi}_\gamma(\bar{x}),\sem{L\phi}_\gamma(\bar{y}))\\
      &=d(\sem{L}(\ftFinK(\Act \boxplus \sem{\phi}_\gamma)(s')), \sem{L}(\ftFinK(\Act \boxplus \sem{\phi}_\gamma)(t'))\\
      &= d(\sem{L}(s), \sem{L}(t))\\
      &>d(s,t),
    \end{split}
  \end{equation*}
  so~$\mathcal{L}$ is not invariant with respect to
  probabilistic trace distance, contradiction.
\end{proof}
\end{remark}

\subsection{Details for \autoref{sec:theories} (\nameref{sec:theories})}

\subsubsection*{Proof of \autoref{prop:theoryMonad}}
  Let $\mathcal{T} = (\Sigma, \delta, E)$ be a depth-$1$ graded
  quantitative equational theory and $\mathbb M$ the graded monad it
  induces.
  \begin{enumerate}[wide]
  \item %
    We show that Diagram~\eqref{coequalizer} is a
    coequalizer. Elements of $M_1M_nX$ are represented as equivalence
    classes~$[s]$ of depth-$1$-terms over depth-$n$ terms over
    variables from~$X$; we briefly refer to such terms~$s$ as
    \emph{layered terms}. More presicely speaking, the element of
    $[s]\in M_1M_nX$ represented by~$s$ is formed by first taking the
    equivalence classes of the lower depth-$n$ terms, and then taking
    the equivalence class of the depth-$1$ tern over $M_nX$ thus
    obtained. We write~$\bar s$ for the collapse of~$s$ into a
    deptn-$n+1$ term over~$X$, so~$\bar s$ represents the element
    $\mu^{10}([s])\in M_{n+1}X$. Given layered terms~$s,t$,
    $d(\mu^{1n}_X([s]), \mu^{1n}_X([t])) = \epsilon$ implies that
    there is a proof for $\bar s =_\epsilon \bar t$ in the
    context~$\Gamma_X=\{x=_\delta y\mid x,y\in X, d(x,y)=\delta\}$.

    Approximate equalities of layered terms in~$M_1M_nX$ are derived
    by regarding the lower depth-$n$ terms as variables, assembled in
    a context~$\Gamma_0$ containing all derivable approximate
    equalities among then, and then applying the usual derivation
    system to the upper depth-$1$ terms. The coequalizer of
    $M_1\mu^{0n}_X$ and $\mu^{10}M_nX$ as in~\eqref{coequalizer} is
    formed by additionally allowing deriviation steps using strict
    equalities where depth-$0$ operations are shifted from the upper
    depth-$1$ layer to the lower depth-$n$ layer. Equivalently (by
    $(\mathbf{nexp})$, which on strict equalities acts like a
    congruence rule), this amounts to using as assumptions, instead of
    just~$\Gamma_0$, all approximate equalities $u=_\epsilon v$ among
    depth-$0$ terms~$u,v$ over depth-$n$ terms over~$X$ (i.e.\
    representatives of elements of $M_0M_nX$) that hold in $M_nX$
    (i.e.\ the elements of $M_0M_nX$ represented by~$u,v$ are
    identified up to~$\epsilon$ by $\mu^{0n}$); that is, we assume, in
    the upper-layer depth-$1$ deriviations, not only the metric
    structure of~$M_nX$ but also the structure of $M_nX$ as an
    $M_0$-algebra. We write~$\Delta$ for the set of these assumptions.

    We now claim that whenever $\bar s=_\epsilon \bar t$ is derivable
    in context~$\Gamma_X$ for layered terms~$s,t$, then
    $s=_\epsilon t$ is derivable, as a depth-$1$ approximate equality,
    from the assumptions in~$\Delta$ (of course, like the usual
    assumptions on variables, the assumptions in~$\Delta$ cannot be
    substituted into). By the above discussion, it follows
    that~\eqref{coequalizer} is indeed a coequalizer. We proceed by
    induction on the derivation of $\bar s=_\epsilon \bar t$,
    distinguishing cases on the last step. We note that
    since~$\bar s$, $\bar t$ have depth at least~$1$, the case for
    $\mathbf{(assn)}$ does not occur. The remaining cases are as
    follows.
    \begin{itemize}
    \item $\mathbf{(refl)}$: In this case, the layered terms $s$
      and~$t$ collapse into syntactically identical terms, which
      implies that they differ only by shifting depth-$0$ operations
      between the upper depth-$1$ layer and the lower depth-$n$
      layer, and hence are derivably equal under~$\Delta$.
    \item Steps using $\textbf{(sym)}$, $\textbf{(triang)}$,
      $\textbf{(wk)}$, and $\textbf{(arch)}$ can just be copied. As an
      example, we treat the case for $\mathbf{(triang)}$ in detail:
      Since approximate equalities derivable in context~$\Gamma_X$ are
      uniform-depth and use only variables from~$\Gamma_X$, the
      intermediate term used in the last step is necessarily a
      depth-$n+1$ term over~$X$, and therefore a collapse~$\bar w$ of
      some layered term~$w$. That is, we have concluded
      $\bar s=_\epsilon \bar t$ from $\bar s=_{\delta_1}\bar w$ and
      $\bar w=_{\delta_2}\bar t$ where
      $\epsilon=\delta_1+\delta_2$. By induction, $s=_{\delta_1}w$ and
      $w=_{\delta_2}t$ are derivable under~$\Delta$, and then the same
      holds for $s=_\epsilon t$.
    \item $\mathbf{(nexp)}$: Nonexpansiveness of an operation~$f$ is
      equivalently phrased as an axiom
      $x_1=_\epsilon y_1,\dots,x_n=_\epsilon y_n\vdash
      f(x_1,\dots,x_n)=_\epsilon f(y_1,\dots,y_n)$, so we omit this
      case, referring to the case for $\mathbf{(ax)}$.
    \item $\mathbf{(ax)}$: In this case, we have applied an axiom
      $\Gamma\vdash u=_\epsilon v$, and derived
        $\sigma(u) =_\epsilon \sigma(v)$ from $\sigma(x)=_\delta\sigma(y)$
      for all $(x=_\delta y)\in\Gamma$. First we consider the case
      that~$u,v$ are depth-$0$. Then for $(x=_\delta y)\in\Gamma$,
      $\sigma(x)$ and $\sigma(y)$ have depth $n+1$, and hence arise by
      collapsing layered terms $w_x$, $w_y$; by induction,
      $w_x=_\delta w_y$ is derivable under~$\Delta$, and again
      applying $\Gamma\vdash u=_\epsilon v$, we obtain that
        $\sigma'(u)=\sigma'(v)$ is derivable under~$\Delta$ where $\sigma'$
      is the substitution given by $\sigma'(x)=w_x$. Since the layered
        terms $\sigma'(u)$ and $\sigma'(v)$ collapse to $\bar s$
      and~$\bar t$, respectively, they differ from~$s$ and~$t$,
      respectively, only by shifting depth-$0$ operations between the
      layers, and hence $s=_\epsilon t$ is derivable under~$\Delta$.

      The remaining case is that $u,v$ have depth~$1$. In this case,
      for $(x=_\delta y)\in\Gamma$, $\sigma(x)$ and $\sigma(y)$ have
      depth~$n$, and $\sigma(x)=_\delta \sigma(y)$ is derivable, hence
      in~$\Delta$. We write $u_\sigma$ and~$v_\sigma$ for the layered
      terms with depth-$1$ layer given by~$u$ and~$v$, respectively,
      and depth-$n$ layer given by~$\sigma$. Then
      $u_\sigma=_\epsilon v_\sigma$ is derivable under~$\Delta$. The
      layered terms~$u_\sigma$ and $v_\sigma$ differ from~$s$ and~$t$,
      respectively, at most by shifting depth-$0$ operations between
      the depth-$1$ layer and the depth-$n$ layer, so $s=_\epsilon t$
      is also derivable under~$\Delta$.
    \end{itemize}
  \item Second, we show that the $\mu^{nk}$ are epi.  The $\mu^{1k}$
    are (regular) epimorphisms since diagram \ref{coequalizer} is a
    coequalizer diagram. This implies that all $\mu^{nk}$ are epi. To
    see this, first note that each $\mu^{0k}$ is split epi by the unit
    law $\mu^{0k}\eta M_k = id_{M_k}$. We show that $\mu^{nk}$ is epi
    by induction on $n$. The base cases $n \le 1$ are already
    discharged. For the step from~$n$ to $n+1$, recall the associative
    law $\mu^{n+1,k}\mu^{1n}M_k = \mu^{1,m+k}M_1\mu^{nk}$. We have
    that $\mu^{1,m+k}$ is epi, and by induction, so is
    $M_1\mu^{nk}$. Thus, both sides of the associative law are epi. It
    follows that $\mu^{(n+1)k}$ is epi as desired.
  \qed
  \end{enumerate}

\subsection{Details on \nameref{sec:presentations}}
Fix $k\colon [0,1]^2 \to [0,1]$ such that

\begin{enumerate}
	\item $x + y \geq z$ and $x' + y' \geq z'$ implies $k(x, x') + k(y, y') \geq k(z,z')$.
	\item $k(x, y) = 0$ if and only if $x = y = 0$
\end{enumerate}

\noindent Put  $d_{A \otimes B}((a, b),(a', b')) = k(d(a, a'), d(b, b'))$.

\begin{lemma}
  The function $d_{A\otimes B}$ defined above is a metric on the set
  $A \times B$.
\end{lemma}
\begin{proof}
	For reflexivity, we have that $d_{A\otimes B}((a,b), (a,b)) = k(d_A(a, a), d_B(b,b)) = k(0, 0) = 0$.
	For positivity, suppose that $d_{A\otimes B}((a,b), (a',b')) = 0$. Then $k(d(a, a'), d(b, b')) = 0$, implying $d(a, a') = 0$ and thus $a = a'$ and similarly $b = b'$, therefore $(a,a') = (b, b')$.
	For the triangle inequality, we have that
	\[
		\begin{split}
			&d_{A\otimes B}((a, b), (a', b')) + d_{A\otimes B}((a', b'), (a'', b''))\\
			&= k(d(a, a'), d(b,b')) + k(d(a', a''), d(b', b''))\\
			&\geq k(d(a, a''), d(b, b''))\\
			&= d_{A\otimes B}((a, b), (a'', b'')).
		\end{split}
 	\]
\end{proof}

\begin{lemma}
	The $p$-norms for $1 \leq p\in \nat$ with discount factor $\delta \in (0, 1]$, i.e. $k(x, y) = \sqrt[\leftroot{-3}\uproot{3}p]{x^p + (\delta y)^p}$, satisfy these conditions.
\end{lemma}
\begin{proof}
  For the second condition: It is clear that $k(0, 0) = 0$. So suppose
  $k(x, y) = 0$. This is the case if and only if
  $x^p+ (\delta y)^p = 0$, which in turn implies $x = y = 0$.  For the
  first condition, suppose that $x + y \geq z$ and $x' + y'\geq z'$. Then
  \[
    \begin{split}
      &k(x, x') + k(y, y')\\
      &= \sqrt[p]{x^p + (\delta y)^p} +\sqrt[p]{x'^p + (\delta y')^p}\\
      &\geq \sqrt[p]{(x + x')^p + ((\delta y) + (\delta y'))^p} \quad \text{(Minkowski inequality)}\\
      &\geq \sqrt[p]{z^p + (\delta z')^p} \quad \text{(Monotonicity of $+$, $\sqrt[p]{-}$ and $\delta\cdot$)}\\
      &= k(z, z')
    \end{split}
  \]
\end{proof}
\begin{lemma}
	The sup norm $k(x, y) = x\vee \delta y$ satisfies these conditions.
\end{lemma}
\begin{proof}
	Satisfaction of the second condition is obvious, we show the first condition. Assume $x + y \geq z$ and $x' + y'\geq z'$. Then $k(x, x') + k(y, y') = (x\vee \delta x') + (y\vee \delta y') \geq (x+y)\vee (\delta x' +\delta y') \geq z\vee \delta z' = k(z, z')$
\end{proof}

\subsubsection*{Proof of \autoref{lem:ManhattanNexp}}
Let $(a, s), (b,t) \in \Act \boxplus TX$ where
$s = g(x_1,\ldots, x_n)$ and $t = h(y_1, \ldots, y_m)$ are terms of
$TX$ written in form of an algebraic equational theory associated to
$T$. Then we have that $d((a,s), (b,t)) = d(a, b) \oplus d(s,
t)$. Therefore
\[
  \begin{split}
    &d(\lambda_X(a, s), \lambda_X(b, t))\\
    & = d(g((a,x_1), \ldots, (a,x_n)), h(((b, y_1), \ldots, (b, y_m))))\\
    & \leq d(g((a,x_1), \ldots, (a,x_n)), g(((b, x_1), \ldots, (b, x_n))))\\
    & \oplus d(g((b,x_1), \ldots, (b,x_n)), h(((b, y_1), \ldots, (b, y_m))))\\
    & \leq d(a, b) \oplus d(g((b,x_1), \ldots, (b,x_n)), h(((b, y_1), \ldots, (b, y_m))))\\
    & = d(a, b) \oplus d(s, t) = d((a,s), (b,t))
  \end{split} \hfill
\qed
\]

\subsubsection*{Proof of \autoref{lem:distributive}}

We proceed in two parts. First, we show that $(M_nX)_{n\in \nat}$ is a
quantitative graded $\mathcal{T}[\Act]$ algebra, i.e. a quantitative
algebra \cite{mardare2016quantitative}, extended with grades in the
obvious way (generalized from the set-based
concept~\cite{DBLP:conf/concur/DorschMS19}). Second, we show that
$(M_nX)_{n\in \nat}$ is generated from~$X$ vie the operations of
$\mathcal{T}[\Act]$. Lastly, we prove completeness of the axioms with
respect to this algebra. We interpret depth-0 operations on
$M_nX=T(\Act^{\otimes n} \otimes X)$, ensuring satisfaction
of~$\mathcal{T}$, using the fact that~$T$ is the monad generated
by~$\mathcal{T}$, so that every set~$TY$ is a free
$\mathcal{T}$-algebra. Depth-1 operations $a$ are given by the
functions
$a_n\colon T(\Act^{\otimes n} \otimes {-}) \to T(\Act^{\otimes
  n+1}\otimes {-})$ defined by
$a_n(s) = \lambda_{\Act^{\otimes n}\otimes X}(a,s)$. For satisfaction
of the distributivity axiom, let
$v_1,\dots,v_n \in T(\Act^{\otimes n} \otimes X)$, represented as
$\mathcal{T}$-terms, and let~$g$ be an $n$-ary operation
of~$\mathcal{T}$; then
$a(g(v_1,\dots,v_n))=\lambda(a,g(v_1,\dots,v_n))= T\langle a,
\id_{\Act^{\otimes n} \otimes X} \rangle(t) = g((a, v_1), \ldots, (a,
v_n))$. For the distance axiom
$x =_\epsilon y \vdash a(x) =_{k(d(a, b), \epsilon)} b(y)$, let
$s, t \in T(\Act^{\otimes n} \otimes {-})$. %
Then we have
\[
  \begin{split}
    d(a(s), b(t)) &= d(\lambda_X(a, s), \lambda_X(b, t))\\
    & \le d_{\Act\otimes T(\Act^{\otimes n}\otimes X)}((a,s),(b,t))\\
    &= k(d(a, b), d(s,t)),
  \end{split}
\]
using nonexpansiveness of~$\lambda$ in the second step.

Generatedness of~$M_nX$ under these operations is then clear. For
completeness, we show that low distances in $M_nX$ are derivable using
quantitative equational reasoning: Let
$s, t \in M_nX=T(\Act^{\otimes m}\otimes X)$. By generatedness, we can
write~$s,t$ as terms $s=g(\sigma_1(x_1), \ldots, \sigma_n(x_n))$,
$t=h(\rho_1( y_1), \ldots, (\rho_m(y_m))$ where $g,h$ are
$\mathcal{T}$-terms, $\sigma_i,\rho_i\in\Act^{\otimes n}$, and
$x_i,y_i\in X$ for $i=1,\dots,n$.  Let~$\Gamma$ contain all distances
among $x_1,\dots,x_n,y_1.\dots,y_n$. Assume $d(s, t) \leq \epsilon$,
then we need to show that $\Gamma \vdash s =_\epsilon t$. By
completeness of the theory~$\mathcal{T}$ for~$T$, it is clear that
$\Gamma' \vdash s =_\epsilon t$ where $\Gamma'$ contains all distances
among the trace-poststate pairs
$\sigma_1(x_1),\dots,\sigma_n(x_n)),\rho_1(y_1),\dots,\rho_n(y_n))$. So
we reduce to proving that distances of these traces can be derived,
i.e.
$\Gamma \vdash (a_1, \ldots, a_n, x) =_{\epsilon'} (b_1, \ldots, b_n,
y)$ if
$d((a_1, \ldots, a_n, x), (b_1, \ldots, b_n, y)) \leq \epsilon'$. This
is immediate by induction over $n$, where the inductive step applies
the distance axiom.
\qed

\subsection{Details for \autoref{sec:logic}\\
(\nameref{sec:logic})}

\subsubsection{Graded Behavioural Distance is Below Branching-Time Behavioural Distance}
  \label{sec:coarser}
  \noindent We next substantiate the claim that every graded
  behavioural distance is below branching-time behavioural distance;
  this needs some forward references to \autorefsecs{sec:logic}
  and~\ref{sec:expr}. In fact, it is even true that every graded
  behavioural distance is below finite-depth branching-time distance
  as summarized in \autoref{expl:branching-time}. This, in turn, is
  immediate from the fact that the behaviour maps as introduced in
  \autoref{sec:logic} factors through the final sequence, as stated
  formally in the next lemma.

\begin{lemma}\label{lem:coarser}
  Let $(\mathbb{M},\alpha)$ be a graded semantics of a functor~$G$ on
  metric spaces. The morphisms $\gamma^{(i)}\colon X \to M_n1$ factor
  through $p_i\colon X \to G^i1$ as defined in \autoref{expl:branching-time}.
\end{lemma}
\begin{proof}
  We show the claim by induction on $i$. For $i = 0$ we have by
  naturality of $\eta$ that
  $\gamma^{(0)} = M_0! \cdot \eta_X = \eta_1 \cdot !_X = \eta_1 \cdot
  p_0$. For the inductive step assume that $\gamma^{(i)} = f_i \cdot
  p_i$. Then we have

  \begin{align*}
    \gamma^{(i+1)} &= \mu^{1n} \cdot M_1\gamma^{(n)} \cdot \alpha
      \cdot \gamma & \by{Def. $\gamma^{(i+1)}$}\\
    &= \mu^{1n} \cdot M_1f_i \cdot M_1p_i \cdot \alpha_X \cdot \gamma
      &\by{IH}\\
    &= \mu^{1n} \cdot M_1f_i \cdot \alpha_{G^iX} \cdot Gp_i \cdot
      \gamma & \by{Naturality of $\alpha$}\\
    &= \mu^{1n} \cdot M_1f_i \cdot \alpha_{G^iX} \cdot p_{i+1}
      & \by{Def. $p_{i+1}$} & \qedhere
  \end{align*}
\end{proof}

\subsubsection*{Proof of \autoref{thm:invariance}}
\label{prf:invariance}

The proof of \autoref{thm:invariance} 
\cite[Proposition 21]{DBLP:journals/corr/abs-2307-14826} is based on the
evaluation $\sem{\phi}_\gamma\colon X \to \Omega$ of a uniform depth-n formula
$\phi$ factoring through $M_n1$ via a map $\sem{\phi}_\mathbb{M}$. Since these facts are used in the proof of \autoref{thm:main}, we recall the details
here.

We define for each uniform depth-$n$ formula~$\phi$ an $M_0$-algebra
homomorphism $\sem{\phi}_{\mathbb M}\colon M_n1 \to \Omega$ and show
that on a coalgebra $(X, \gamma)$, we have
$\sem{\phi}_\gamma = \sem{\phi}_\mathbb{M} \cdot \gamma^{(n)}$.  The
claim then follows from the fact that the $\sem{\phi}_\mathbb{M}$ are
nonexpansive and
$d^\alpha(x, y) \geq d_{M_n1}(\gamma^{(n)}(x), \gamma^{(n)}(y))$ for
$x, y \in X$. We define $\sem{\phi}_\mathbb{M}$ recursively as
follows.

\begin{itemize}
\item  $\sem{c}_\mathbb{M} = M_01 \xrightarrow{M_0\hat{c}} M_0\Omega \xrightarrow{o} \Omega$ for $c \in \Theta$;
\item $\sem{ p(\phi_1, \ldots, \phi_k)}_\mathbb{M} = \sem{p} \cdot \langle \sem{ \phi_1}_\mathbb{M}, \ldots, \sem{ \phi_k }_\mathbb{M} \rangle$ for $p \in \mathcal{O}$ $k$-ary;
\item $\sem{ L \phi}_\mathbb{M} = \rsem{L}^\bullet(\sem{\phi'}_\mathbb{M})$ for
  $L \in \Lambda$ and $\sem{L} = \rsem{L} \cdot \alpha_\Omega$ as
  per \autoref{def:graded-logic}.
\end{itemize}
Here, the definition of modal operators is by freeness of the canonical algebra, that is,
$\rsem{L}^\bullet(\sem{\phi'}_\mathbb{M})$ is the unique morphism that makes the following square commute:
\begin{equation}\label{eq:modal}
  \begin{tikzcd}
    M_1M_n1 \arrow[r, "M_1\sem{\phi'}_\mathbb{M}"] \arrow[d,
    "\mu^{1n}"] & M_1\Omega \arrow[d, "\rsem L"] \\
    M_{n+1}1 \arrow[r, "\rsem{L}^\bullet(\sem{\phi'}_\mathbb{M})"]                      & \Omega
\end{tikzcd}
\end{equation}
It is straightforward to show by induction on the depth of $\phi$ that
the morphism $\sem{\phi}_\mathbb{M}$ defines a homomorphism of
$M_0$-algebras $(M_n1, \mu^{0n})$ and $(\Omega, o)$, which is needed
for $\sem{L\phi}_\mathbb{M}$ to be defined.

Given a coalgebra $(X, \gamma)$ and a depth-$n$ formula $\phi$ of
$\mathcal L$, we now show by induction on $\phi$ that
$\sem{\phi}_\gamma = \sem{\phi}_\mathbb{M} \cdot \gamma^{(n)}$.

For the case of $\phi = c \in \Theta$ we have, by unfolding
definitions, that $\sem{c}_\gamma = \hat{c} \cdot !_X$ and
$ \sem{\phi}_\mathbb{M} \cdot \gamma^{(n)} = o\cdot M_0\hat{c} \cdot
M_0!_X \cdot \eta_X$, which are the outer paths in the following
diagram:
\begin{equation*}
  \begin{tikzcd}
    X \arrow[r, "!"] \arrow[d, "\eta_X"] & 1 \arrow[d, "\eta_1"] \arrow[r, "\hat c"] & \Omega \arrow[d, "\eta_\Omega"] \arrow[rd, "id"] &        \\
    M_0X \arrow[r, "M_0!"]               & M_01 \arrow[r, "M_0 \hat c"]              & M_0\Omega \arrow[r, "o"]                         & \Omega
  \end{tikzcd}
\end{equation*}
The squares commute by  naturality of $\eta$, while commutativity of the triangle is implied by $o$ being an $M_0$-algebra.

The step for formulae of the form $\phi = p(\phi_1, \ldots, \phi_n)$
is immediate from the definitions. For $\phi = L\phi'$ with
$\phi'$ being of uniform depth~$n$, we have

\begin{equation*}
	\begin{alignedat}[t]{2}
	  &\sem{\phi}_\gamma = \sem{L} \cdot G\sem{\phi'}_\gamma \cdot \gamma\\
	  &= \rsem L \cdot \alpha_\Omega \cdot G\sem{\phi'}_\gamma \cdot \gamma \\
		&= \rsem L \cdot M_1\sem{\phi'}_\gamma \cdot\alpha_X\cdot \gamma & \by{naturality of $\alpha$}\\
		&= \rsem L \cdot M_1\sem{\phi'}_\mathbb{M}\cdot M_1\gamma^{(n)}\cdot\alpha_X\cdot \gamma & \by{induction}\\
    &=\rsem{L}^\bullet(\sem{\phi'}_\mathbb{M}) \cdot \mu^{1n} \cdot M_1\gamma^{(n)}\cdot\alpha_X\cdot \gamma & \by{\ref{eq:modal}}\\
		&= \sem{\phi}_\mathbb{M} \cdot \gamma^{(n+1)} & \by{definitions}\\
  \end{alignedat}
\end{equation*}
\qed

\subsection{Details for \autoref{sec:expr} (\nameref{sec:expr})}
\subsubsection*{Details for \autoref{rem:density}}\jfnote{should we mention that this holds in several examples, particularly in the 2-valued case?}
We make explicit the fact that given an initial family $\mathfrak{A}$ of $\Pmet$-homomorphisms (i.e.\ quantitative join semilattice homomorphisms)  $A\to [0,1]$, the closure of $\mathfrak{A}$ under all admissible propositional operators (i.e.\ $\Pmet$-homomorphisms $[0,1]^k\to[0,1]$) may not be dense in the space of all $\Pmet$-homomorphisms $A\to [0,1]$, even when~$A$ is free.
Take for instance the $\Pmet$-algebra $\Pmet (X, d_X)$ where $X = \{a, b, c\}$ with $d(a,b) = d(b,c) = 0.3$ and $d(a,c) = 0.4$.
	We attempt to recover the $\Pmet$-homomorphism $f\colon \Pmet (X, d_X) \rightarrow [0,1]$ uniquely defined by $f(\{a\}) = 0.6$, $f(\{b\}) = 0.8$, $f(\{c\}) = 1$ and $f(\{d\}) = 0$.
		We define a set $\mathfrak{A}$ as above as consisting of all join semilattice homomorphisms $g$ where either a) $g \not \leq f$ or b) $g(\{c\}) < 1$. This set of functions is initial.
We show that $\mathfrak{A}$ is also closed under the propositional
operators mentioned above: Let $o\colon [0, 1]^n \rightarrow [0, 1]$
be a join semilattice homomorphism, and let
$g_1,\dots,g_n\in\mathfrak A$. We have to show that
$o\langle g_i\rangle\in\mathfrak A$. It is clear that
$o\langle g_i\rangle$ is a join semilattice morphism. For the
remaining property, we distinguish cases as follows.
\begin{itemize}
\item Suppose that $o(1,\dots,1) < 1$. Then by monotonicity of join
  semilattice morphisms,
  $o(\langle g_i \rangle_{i \leq n}(\{c\})) < 1$, so
  $o\langle g_i \rangle\in\mathfrak A$.
\item Otherwise, $o(1^n) = 1$. By join continuity, the
  set~$J=\{j\in\{1,\dots,n\}\mid o(0^{j-1}10^{n-j-1}) = 1\}$ is
  nonempty. Since $o$ preserves the empty join, i.e. $o(0^n) = 0$, we
  know by nonexpansiveness that
  \begin{equation}\label{eq:j-ident}
    o(0^{j-1}v0^{n-j-1}) = v\quad\text{for all
      $j\in J$.}
  \end{equation}
  If $g_j(\{c\}) < 1$ for all $j \in J$, then
  $o(\langle g_i \rangle_{i \leq n}(\{c\})) < 1$ by join continuity,
  so $o\langle g_i \rangle\in\mathfrak A$ as claimed. Otherwise, we
  have $j \in J$ such that $g_j(\{c\}) = 1$. Since
  $g_j\in\mathfrak A$, this means that we have $z \in X$ such that
  $g_j(z) > f(z)$, implying by~\eqref{eq:j-ident} that
  $o(0^{j-1}g_j(z)0^{n-j-1}) = g_j(z) > f(z)$ and thus
  $o(\langle g_i \rangle_{i \leq n}(z)) > f(z)$ by monotonicity. Thus,
  $o\langle g_i \rangle\not\le f$, so
  $o\langle g_i \rangle\in\mathfrak A$.
\end{itemize}
On the other hand, it is clear that $\mathfrak{A}$ is not dense in
$\Galg{0}{M}(\Pmet(X), [0,1])$.

\subsubsection*{Proof of \autoref{thm:main}}

  Since uniform depth-$n$ formulae $\phi$ factor through $M_n1$ via
  algebra homomorphisms
  $\sem{\phi}_\mathbb{M} \colon M_n1 \to \Omega$, as
  shown in the \nameref{prf:invariance}, we only have to show
  that the closure of the family of maps
  \begin{equation*}
    \{\llbracket \phi \rrbracket_\mathbb{M}: M_n1 \to [0, 1] \mid \phi \text{ is
      a depth-$n$ formula}\}
  \end{equation*}
  under the propositional operators in $\mathcal{O}$ has property
  $\Phi$ and is thus initial for each~$n$.  We proceed by
  induction on $n$. The base case $n = 0$ is immediate by $\Phi$-type
  depth-0 separation. For the inductive step let $\mathfrak{A}$ denote
  the set of evaluations $M_n1 \rightarrow \Omega$ of depth-$n$
  formulae. By the induction hypothesis, $\mathfrak{A}$ has property
  $\Phi$. By definition,~$\mathfrak{A}$ is closed under propositional
  operators in
  $\mathcal{O}$. %
  By $\Phi$-type depth-$1$ separation, it follows that set
\begin{equation*}
  \{\rsem{L}^\bullet(\llbracket \phi \rrbracket_\mathbb{M}) \mid L \in
  \Lambda,\; \phi \text{ a depth-$n$ formula}\}
\end{equation*}
has property $\Phi$, proving the claim.
\qed

\subsubsection*{Details on \autoref{expl:separation}.1 (Metric Streams)}
In the following, we write $\mathcal{L}^{\stream}$ for the logic of
metric streams featuring only the truth constant~$1$ and
modalities~$\modal{a}$ for $a\in\Act$.

\begin{lemma}\label{lem:metric-streams-normed}
	$\mathcal{L}^{\stream}$ is normed isometric depth-0 and normed isometric depth-1 separating.
\end{lemma}
\begin{proof}
  For normed isometric depth-$0$ separation, note that $M_01 = 1$ and
  that $\llbracket \top \rrbracket = M_0\hat{\top}$ is defined by
  $* \mapsto 1$, so that $\{\llbracket \top \rrbracket\}$ is normed
  isometric. For normed isometric depth-$1$ separation, note that
  since $M_0=\id$ in the present case, canonical $M_1$-algebras have
  the form $\id\colon GA_0\to A_1=GA_0=\Act\times A_0$. Let~$A$ be
  such a canonical $M_1$-algebra, and let $\mathfrak{A}$ be a normed
  isometric set of morphisms $A_0 \rightarrow \Omega$. Let
  $(a, v), (b, w) \in M_1A_0=\Act\times A_0$.  Since~$\mathfrak{A}$ is
  normed isometric, there is $f \in \mathfrak{A}$ such that $f(v) = 1$
  and $f(w) = 1- d(v, w)$.  To show that $\Lambda(\mathfrak{A})$ is
  normed isometric, we show that
  $\sem{\modal{a}}^\bullet(f)\in\Lambda(\mathfrak A)$ exhibits the required
  properties. On the one hand, we have
  \begin{equation*}
    \begin{split}
      &\llbracket \modal{a} \rrbracket^\bullet(f)((a, v))\\
      & = \llbracket \modal{a} \rrbracket((a, f(v))\\
      & = \min\{1- d(a, a), f(v)\}\\
      & = f(v)\\
      & =  1
    \end{split}
  \end{equation*}
  and on the other hand,
  \begin{equation*}
    \begin{split}
      &\llbracket \modal{a} \rrbracket^\bullet(f)((b, w))\\
      & = \llbracket \modal{a} \rrbracket((b, f(w)))\\
      & = \min\{1- d(a, b), f(w)\}\\
      & = \min\{1- d(a, b), 1- d(v, w)\}\\
      & = 1 - \max\{d(a, b), d(v, w)\}\\
      & = 1- d((a, v), (b, w)).
    \end{split}
  \end{equation*}
\end{proof}

\begin{proposition}
  $\mathcal{L}^{\stream}$ is not initial-type depth-1 separating.
\end{proposition}
\begin{proof}
  Let $\Act = \{a, b\}$ with $d(a, b) = 0.8$. Let $A$ be a canonical
  $M_1$-algebra (described as in the proof of
  \autoref{lem:metric-streams-normed}) with $A_0 = \{v, w\}$ where
  $d(v, w) = 0.5$. The map $f\colon A_0\to \Omega$ defined by
  $f(w) = 0.25$ and $f(v) = 0.75$ is thus initial; we show that
  $\Lambda(\{f\})=\{\sem{\modal{a}}^\bullet(f),\sem{\modal{b}}^\bullet(f)\}$ (a
  set of maps $A_1=\Act\times A_0\to\Omega$) is not, thus proving the
  claim (recall that there are no propositional operators). For
  $\sem{\modal{a}}^\bullet(f)$, we have
  \begin{equation*}
    \begin{split}
      &|\llbracket \modal{a} \rrbracket^\bullet(f)((a, v)) - \llbracket \modal{a} \rrbracket^\bullet(f)((b, w))|\\
      & =|\llbracket \modal{a} \rrbracket((a, f(v))) - \llbracket \modal{a} \rrbracket((b, f(w))|\\
      & =|\min\{1- d(a, a), f(v)\} - \min\{1- d(a, b), f(w)\}|\\
      & =|f(v) - \min\{1- d(a, b), f(w)\}|\\
      & =|0.75 - \min\{0.2, 0.25\}| = 0.55\\
    \end{split}
  \end{equation*}
  Moreover, for $\sem{\modal{b}}^\bullet(f)$, we have
  \begin{equation*}
    \begin{split}
      &|\llbracket \modal{b} \rrbracket^\bullet(f)((a, v)) - \llbracket \modal{b} \rrbracket^\bullet(f)((b, w))|\\
      & =|\llbracket \modal{b} \rrbracket((a, f(v))) - \llbracket \modal{b} \rrbracket((b, f(w))|\\
      & =|\min\{1- d(b, a), f(v)\} - \min\{1- d(b, b), f(w)\}|\\
      & =|\min\{1- d(b, a), f(v)\} - f(w)\}|\\
      & =|\min\{0.2, 0.75\} - 0.25\}| = 0.05.
    \end{split}
  \end{equation*}
  Since $d((a,v), (b,w)) = 0.8$, this shows that
  $\Lambda(\mathfrak{A})$ is not initial.
\end{proof}

\subsubsection*{Details for \autoref{expl:separation}.2 (Metric Transition Systems)}
\begin{lemma}
  The graded quantitative theory in
  \autoref{expl:trace-theories} induces the
  graded monad $\Pmet(L^n \times {-})$
\end{lemma}
\begin{proof}
  By \autoref{lem:distributive}, it suffices to show that the natural
  transformation
  $\lambda\colon \Act \times \Pmet{-} \to \Pmet(\Act \times {-})$
  given by $\lambda_X(a, S) = \{(a, s) \mid s\in S\}$ is
  nonexpansive. Let $(a, S), (b, T) \in \Act \times \Pmet X$ for a
  metric space~$X$. Then
  \begin{equation*}
    \begin{split}
      &d(\lambda_X(a, S), \lambda_X(b, T))\\
      &= \big(\bigvee_{s\in S} \bigwedge_{t\in T}d((a, s), (b, t))\big)\vee\big( \bigvee_{t\in T} \bigwedge_{s\in S}d((a, s), (b, t))\big)\\
      &= \big(\bigvee_{s\in S} \bigwedge_{t\in T}d(a, b) \vee d(s, t)\big)\vee\big( \bigvee_{t\in T} \bigwedge_{s\in S}d(a,b) \vee d(s, t)\big)\\
      &= \big(\bigvee_{s\in S} \big(d(a, b) \vee \bigwedge_{t\in T} d(s, t)\big)\big)\vee\big( \bigvee_{t\in T} \big(d(a,b) \vee \bigwedge_{s\in S} d(s, t)\big)\big)\\
      &\le d(a, b) \vee \big(\bigvee_{s\in S} \bigwedge_{t\in T}d(s, t)\big)\vee\big( \bigvee_{t\in T} \bigwedge_{s\in S}d(s, t)\big)\\
      &= d((a, S), (b, T)).
    \end{split}
  \end{equation*}
\end{proof}

\begin{lemma}
  The logic for trace semantics given in
  \autoref{expl:graded-logics} is
  expressive.
\end{lemma}
\begin{proof}
  As indicated in \autoref{expl:separation}.\ref{expl:mts-separation},
  we use \autoref{thm:main}, with normed isometry as the initiality
  invariant.  Depth-0 separation is straightforward. For depth-1
  separation, we proceed according to
  \autoref{rem:restrict-separation} and exploit that canonical
  $M_1$-algebras of the form~$M_n1$ are free as $M_0$-algebras, which
  in the case at hand means they are quantitative join semilattices
  carrying the Hausdorff metric. That is, we are given a canonical
  $M_1$-algebra~$A$ such that $A_0=\Pmet(X)$ where~$X$ is a metric
  space, and then $A_1=\Pmet(\Act\times X)$. Let $\mathfrak{A}$ be a
  normed isometric set of nonexpansive join-semilattice morphisms
  $A\to\Omega$. We have to show that $\Lambda(\mathfrak{A})$ is normed
  isometric. So let $S,T\in\Pmet(\Act\times X)$ such that
  $d(S,T)>\epsilon$. Without loss of generality, this means that we
  have $(a,x)\in S$ such that $d(\{(a,x)\},T)>\epsilon$. Put
  \begin{equation*}
    C_{a,\epsilon}(T)=\{y\in X\mid \exists (b,y)\in T.\,d(a,b)\le\epsilon\}.
  \end{equation*}
  Then $d(\{x\},C_{a,\epsilon}(T))>\epsilon$ in~$A_0=\Pmet(X)$. To see
  this, let $y\in C_{a,\epsilon}(T)$, i.e.\ $(b,y)\in T$ for some~$b$
  such that $d(a,b)\le\epsilon$. Then $d((a,x),(b,y))>\epsilon$
  because $d((a,x),T)>\epsilon$, and hence $d(x,y)>\epsilon$, as
  required.

  It follows that
  $d(\{x\}\cup C_{a,\epsilon}(T),C_{a,\epsilon}(T))>\epsilon$.
  Since~$\mathfrak{A}$ is normed isometric, we thus have
  $f\in \mathfrak{A}$ such that
  \begin{equation*}
    |f(\{x\}\cup C_{a,\epsilon}(T))-f(C_{a,\epsilon}(T))|>\epsilon
  \end{equation*}
  and moreover
  $f(\{x\}\cup
  C_{a,\epsilon}(T))\vee f(C_{a,\epsilon}(T))=1$. Since~$f$, being a join
  semilattice morphism, is monotone w.r.t.\ set inclusion, this
  implies that $f(\{x\}\cup C_{a,\epsilon}(T))=1$ and
  $f(C_{a,\epsilon}(T))<1-\epsilon$, and moreover
  $f(\{x\}\cup C_{a,\epsilon}(T))=f(x)\vee f(C_{a,\epsilon}(T))$, so
  $f(\{x\})=1$. It follows that 
        $\modal{a}(f)(S) = 1$.
  It remains to
  show that $\modal{a}(f)(T)<1-\epsilon$. So let $(b,y)\in T$; we
  have to show that $f(y)\wedge (1-d(a,b))<1-\epsilon$. We distinguish
  cases:
  \begin{itemize}
  \item $d(a,b)\le\epsilon$: Then $y\in C_{a,\epsilon}(T)$. Since
    $f(C_{a,\epsilon}(T))<1-\epsilon$ and~$f$ preserves joins, we have
    $f(y)<1-\epsilon$, which implies the claim.
  \item $d(a,b)>\epsilon$: Then $1-d(a,b)<1-\epsilon$, which implies
    the claim. \qedhere
  \end{itemize}
\end{proof}

\subsection{Details for \autoref{sec:examples}
(\nameref{sec:examples})}

\subsubsection*{The Fuzzy Finite Powerset Monad}\label{sec:fuzzy-monad} The fuzzy finite powerset monad $\FPow_\omega$ on $\SET$ is given as
follows\lsnote{Use semimodule terminology}. For a set~$X$, $\FPow_\omega X$ is the set of maps
$X\to[0,1]$ with finite support. For a map~$f\colon X\to Y$,
$\FPow_\omega f$ takes fuzzy direct images; explicitly, for
$S\in\FPow_\omega X$ and $y\in Y$,
$\FPow_\omega f(S)(y)=\bigvee_{x\in X\mid f(x)=y}S(x)$. The unit
$\eta$ takes singletons; explicitly, $\eta_X(x)(x)=1$, and
$\eta_X(x)(y)=0$ for $y\neq x$. Finally, the multiplication~$\mu$
takes fuzzy big unions; explicitly,
$\mu_X(\mathfrak S)(x)=\bigvee_{S\in\FPow_\omega(X)}\mathfrak
S(S)\land S(x)$.

We lift~$\FPow_\omega$ to a functor $\FHaus_\omega$ on metric spaces by means of the
Kantorovich lifting
construction~\cite{bbkk:trace-metrics-functor-lifting}, applied to the
evaluation function $\ev_\Diamond\colon \FPow_\omega[0,1]\to[0,1]$
given by
\begin{equation*}
  \ev_\Diamond (S)=\bigvee_{v\in [0,1]}S(v)\land v,
\end{equation*}
equivalently presented as the predicate
lifting~$\Diamond\colon [0,1]^{-}\to [0,1]^{\FPow_\omega -}$ given by
\begin{equation*}
  \Diamond_X(f)(S)=\bigvee_{x\in X}S(x)\land f(x)
\end{equation*}
for $f\in[0,1]^X$ and $S\in\FPow_\omega(X)$. Explicitly, for a metric
space~$X=(X,d)$, the metric on $\FHaus_\omega X$ is thus given by
\begin{equation}\label{eq:fuzzy-kant}
  d(S,T)=\bigvee_{f\in\Met( X, [0,1])}|\Diamond_X(f)(S)- \Diamond_X(f)(T)|.
\end{equation}
This metric has an equivalent Hausdorff-style characterization~\cite[Example~5.3.1]{DBLP:conf/fossacs/WildS21} as
\begin{multline}\label{eq:fuzzy-haus}
  d(S,T)=\Big(\bigvee_{x}\bigwedge_{y} (S(x)\ominus T(y))\vee (S(x)\wedge d(x,y))\Big)\vee\\
  \Big(\bigvee_{y}\bigwedge_{x} (T(y)\ominus S(x))\vee (T(y)\wedge d(x,y))\Big)
\end{multline}
(in \emph{loc.\ cit.}, only the asymmetric case is mentioned, but
note~\cite[Lemma~5.10]{DBLP:journals/lmcs/WildS22}\footnote{In more detail, one applies \cite[Lemma~5.10]{DBLP:conf/fossacs/WildS21}
  to obtain that the Kantorovich lax
  extension~\cite{DBLP:conf/fossacs/WildS21} agrees with the
  Kantorovich lifting~\cite{bbkk:trace-metrics-functor-lifting} on
  pseudometrics when we additionally consider the dual modality~$\Box$
  of~$\Diamond$; this Kantorovich lifting, in turn, is easily seen to
  coincide with~\eqref{eq:fuzzy-kant}. On the other hand, applying the
  characterization given
  in~\cite[Example~5.3.1]{DBLP:conf/fossacs/WildS21} both
  to~$\Diamond$ and, dualizing appropriately, to~$\Box$ in the
  Kantorovich lax extension yields~\eqref{eq:fuzzy-haus}.}).

This metric in fact lifts the monad~$\FPow_\omega$ to a monad
on~$\Met$: The unit is clearly nonexpansive. Nonexpansiveness of the
multiplication, i.e.\ of fuzzy big union, is seen by means of the
Kantorovich description~\ref{eq:fuzzy-kant}; in fact, it suffices to
show that every map
$\Diamond_X(f)\colon\FPow_\omega X\to[0,1]\cong\FPow_\omega1$ is a
morphism of $\FPow_\omega$-algebras. Indeed, we can then argue as
follows: By definition, the family of all $\Diamond_X(f)$ is initial, so
it suffices to show that $\Diamond_X(f)\mu_X$ is nonexpansive for
all~$f$; but $\FPow_\omega$-homomorphy of $\Diamond_X(f)$ means that
these maps equal $\mu_1\FPow_\omega\Diamond(f)$. The
map~$\FPow_\omega \Diamond(f)$ is nonexpansive because $\Diamond(f)$
is nonexpansive and~$\FHaus_\omega$ lifts~$\FPow_\omega$;
nonexpansiveness of $\mu_1$ follows from the fact that join and meet
are nonexpansive operations on $[0,1]$.

So we check that $\Diamond_X(f)$ is a morphism of
$\FPow_\omega$-algebras. Note that under the isomorphism
$\FPow_\omega 1\cong[0,1]$, $\mu_1$ corresponds to $\ev_\Diamond$,
and under the usual Yoneda
correspondence~\cite{DBLP:journals/tcs/Schroder08} between
$\ev_\Diamond$ and
$\Diamond$, %
$\Diamond_X(f)=\ev_\Diamond\FPow_\omega f$. But $\FPow_\omega f$ is a
morphism of free $\FPow$-algebras, and so is $\ev_\Diamond$,
since~$\mu_1$ is.

We claim that $\FHaus_\omega$ is presented by the following
quantitative algebraic theory. We have unary operations~$r$ for all
$r\in[0,1]$, a constant~$0$, and a binary operation~$+$. These are
subject to strict equations saying that $0,+$ form a join semilattice
structure; strict equations stating that the operations~$r$ form an
action of $([0,1],\wedge)$, with~$0$ as zero element:
\begin{equation*}
  1(x)=x\quad r(s(x))=(r\wedge s)(x)\quad 0(x)=0;
\end{equation*}
and the following axioms governing distance:
\begin{equation*}
  x=_\epsilon y\vdash r(x)=_\epsilon s(y)
\end{equation*}
under the side condition that $|r-s|\le\epsilon$. Using
nonexpansiveness of joins, one derives
\begin{equation}\label{eq:fuzzy-eqs}
  x_1=_\epsilon y_1,\dots,x_n=_\epsilon y_n\vdash \sum_i r_i(x_i) =_\epsilon \sum_i s_i (y_i),
\end{equation}
again under the side condition that~$|r_i-s_i|\le\epsilon$ for all~$i$,
where we use big sum notation in the evident sense. %

It is clear that the operations and strict equations induce
$\FPow_\omega$. To see that the full theory induces $\FHaus_\omega$, it remains to
show two things:

\emph{Axiom~\eqref{eq:fuzzy-eqs} is sound over~$\FHaus_\omega X$:}
Since we have already shown that $\FHaus_\omega$ is a monad, it
suffices to show that~\eqref{eq:fuzzy-eqs} holds
in~$\FHaus \{x_1,\dots,x_n,y_1,\dots,y_n\}$ for
$x_1,\dots,x_n,y_1,\dots,y_n\in X$ (with every variable interpreted by
itself, and distances inherited from~$X$). But this is immediate from
the Hausdorff description of the distance: By symmetry, it suffices to
show that the left-hand term in the maximum in~\eqref{eq:fuzzy-haus}
is at most~$\epsilon$. This follows from the fact that by the side
conditions of~\eqref{eq:fuzzy-eqs},
$(r_i\ominus s_i)\vee (r_i\land d(x_i,y_i))\le(r_i\ominus s_i)\vee
d(x_i,y_i)\le\epsilon$ for all~$i$.

\emph{Completeness: Low distances in~$\FHaus_\omega X$ are derivable.}
Let $S,T\in\FHaus_\omega X$ such that $d(S,T)\le\epsilon$; read $S,T$
as terms $S=\sum_xS(x)(x)$, $T=\sum_xT(x)(x)$. We have to show that
$\Gamma\vdash S=_\epsilon T$ is derivable where $\Gamma$ records the
distances of all elements of the (finite) supports of~$S$ and~$T$. Now
$d(S,T)\le\epsilon$ implies that
$\bigvee_x\bigwedge_y (S(x)\ominus T(y))\vee (S(x)\wedge
d(x,y))\le\epsilon$. That is, for every~$x$ there exists $y_x$ such
that $(S(x)\ominus T(y_x))\vee (S(x)\wedge d(x,y_x))\le\epsilon$. This
means that
\begin{enumerate}
\item $S(x)\le T(y_x)\oplus\epsilon$, and
\item either $d(x,y_x)\le\epsilon$ or $S(x)\le\epsilon$. In the latter
  case, we may take $y_x=x$, ensuring both
  $S(x)\le T(y_x)\oplus\epsilon$ and, again, $d(x,y_x)\le\epsilon$.
\end{enumerate}
By the strict equational laws, it follows that
\begin{equation}\label{eq:S-T-epsilon}
  \Gamma\vdash S=\sum_{d(x,y)\le\epsilon}(S(x)\land (T(y)\oplus\epsilon))(x).
\end{equation}
Analogously, we obtain
\begin{equation}\label{eq:S-epsilon-T}
  \Gamma\vdash T=\sum_{d(x,y)\le\epsilon}((S(x)\oplus\epsilon)\land T(y))(y).
\end{equation}
Using~\eqref{eq:fuzzy-eqs} and equations~\eqref{eq:S-T-epsilon}
and~\eqref{eq:S-epsilon-T}, we obtain that
$\Gamma\vdash S=_\epsilon T$ is derivable as desired.

\subsubsection*{Graded Logics for Fuzzy Metric Trace Semantics}
\label{sec:fuzzy-logic} To see that the given logic is indeed a
  graded logic, we have to
show that $\sem{\modal{a}^c}\colon G[0,1]=M_1[0,1]\to[0,1]$ is an
$M_1$-algebra. In terms of the algebraic description of~$\mathbb{M}$,
this means that $\sem{\modal{a}^c}$ should represent term evaluation
w.r.t.\ the $M_0$-algebra structure of $[0,1]$ (which consists in the
usual join semilattice structure of $[0,1]$ and additionally the
interpretation of unary operations~$r$, for $r\in[0,1]$, as taking
meets with~$r$) and suitably chosen nonexpansive interpretations
$b^{\modal{a}^c}\colon [0,1]\to[0,1]$ of the unary depth-1
operations~$b$ for $b\in\Act$; we take $b^{\modal{a}^c}$ to be the map $v \mapsto v \wedge (c - d(a,b))$ on~$[0,1]$ for all~$a$. For depth-0 operations, the fact that
$\sem{\modal{a}^c}$ agrees with term evaluation then just amounts to
the homomorphy condition, which is checked straightforwardly. For the
depth-1 operation~$a$, we calculate as follows (exploiting that every
depth-0 term over $\Omega=[0,1]$ can be normalized to have the form
$\sum r_i(v_i)$ with $r_i,v_i\in[0,1]$):
\begin{align*}
  & \sem{\modal{a}^c}(b(\textstyle\sum r_i(v_i)))\\
  & = \sem{\modal{a}^c}(\textstyle\sum r_i(b(v_i))) &\by{equivalence of depth-1 terms}\\
    & = \textstyle\bigvee_i r_i\wedge v_i \wedge (c - d(a,b)) & \by{definition of $\sem{\modal{a}^c}$}\\
  & = b^{\modal{a}^c}(\textstyle\bigvee_i r_i\wedge v_i), &
\end{align*}
and $\bigvee_i r_i\wedge v_i$ is the evaluation of the term
$\sum_i r_i\wedge v_i$ in $[0,1]$.

It is straightforward to check that the logic is initial-type depth-0
separating. In the proof that it is also initial-type depth-1
separating, we extend the logic to include modal operators
$\modal{a}^c$ for all $c\in[0,1]$; the claim for the original logic,
where we restrict to rational~$c$, then follows immediately from the
fact that the semantics of~$\modal{a}^c$ depends nonexpansively
on~$c$. So let $A$ be a canonical $M_1$-algebra with free $0$-part
$A_0=\FHaus X$ (so $A_1=\FHaus(\Act\times X)$), and let $\mathfrak A$
be an initial family of maps $A_0\to [0,1]$. We have to show that the
family $\Lambda(\mathfrak A)$ of maps $A_1\to [0,1]$ is again
initial. So let $S,T\in\FHaus_\omega(\Act\times X)$ such that
$d(S,T)\ge\epsilon$. W.l.o.g.\ this means that there exist $a,x$ such
that
\begin{equation}\label{eq:fuzzy-ge-eps}
  \bigwedge_{b,y}(S(a,x)\ominus T(b,y))\vee (S(a,x)\wedge
  d((a,x),(b,y)))\ge\epsilon.
\end{equation}
Now define $T_{a,\epsilon}\in\FHaus_\omega X$ by
\begin{equation*}
    T_{a,\epsilon}=\sum_{y, d(a,b) \leq \epsilon} T(b,y)(y).
\end{equation*}
By~\eqref{eq:fuzzy-ge-eps}, we then have
\begin{equation}
    d(S(a,x)(x),T_{a,\epsilon})\ge\epsilon.
\end{equation}
By the Hausdorff-style description of the metric on
$\FHaus(\Act\times X)$, it is immediate that
$d(S(a,x)(x)+T_{a,\epsilon},T_{a,\epsilon})\ge d(S(a,x)(x),T_{a,\epsilon})$, so
\begin{equation}
    d(S(a,x)(x)+T_{a,\epsilon},T_{a,\epsilon})\ge\epsilon.
\end{equation}\jfnote{Justify more clearly}
Since~$\mathfrak A$ is an initial family and the maps in~$\mathfrak A$
are $M_0$-morphisms, in particular monotone w.r.t.\ the ordering of
$\FHaus_\omega X$ as a join semilattice, it follows that we have
$f\in\mathfrak A$ such that $f(T_{a,\epsilon})\le c-\epsilon$ where
$c=f(S(a,x)(x)+T_{a,\epsilon})$. Since~$f$ is an $M_0$-morphism, it follows that
$f(S(a,x)(x))=c$, and $f(S(a,x)(x))=S(a,x)\wedge f(x)$, so

\begin{equation*}
    \appModal{\modal{a}^c}{f}(S)\ge S(a,x)\wedge f(x) \wedge (c- d(a,a))=c.
\end{equation*}
We are thus done once we show that
$\appModal{\modal{a}^c}{f}(T)\le c-\epsilon$. By definition of
$\appModal{\modal{a}^c}{f}(T)$, this means we have to show for
$b\in \Act$ and $y\in X$ that
$T(b,y)\wedge f(y) \wedge (c - d(a,b))\le c-\epsilon$.  We distinguish
cases on $d(a,b)$:
  \begin{itemize}
  \item $d(a,b)\le\epsilon$: Then $T(b,y) \leq T_{a,\epsilon}(y)$, and
    hence $T(b,y)(y)\le T_{a,\epsilon}$. Since
    $f(T_{a,\epsilon})\leq c-\epsilon$ and~$f$ is an $M_0$-morphism,
    in particular is monotone, we have
    $T(b,y)\wedge f(y) = f(T(b,y)(y))\le f(T_{a,\epsilon})\le
    c-\epsilon$, which implies the claim.
  \item $d(a,b)>\epsilon$: Then $c-d(a,b)<c-\epsilon$, which implies
    the claim. 
  \end{itemize}

  \subparagraph*{Failure of expressivity of the logic with only
    $\modal{a}^1$} Consider the
  $\FHaus_\omega(\Act \times {-})$-coalgebra $(X,\gamma)$ where
  $X = \{x, y, z\}$, $\gamma(x) = 0.5(a(z))$, $\gamma(y) = 0.5(b(z))$,
  and $\gamma(z) = 0$, where $d(a,b) = 0.5$. The fuzzy metric trace
  distance of $x,y$ is $0.25$. However, we have
  $\sem{\modal{a}^11}(x) = 0.5 =\sem{\modal{a}^11}(y)$, and similarly
  for $\modal{b}^1$, so the logical distance of $x$ and $y$ in the
  logic with only $\modal{a}^1$, $\modal{b}^1$ is $0$.  To witness the
  distance of $0.25$ logically, we require the modal operator
  $\modal{a}^{0.5}$; indeed, $\sem{\modal{a}^{0.5}1}(x)=0.5$, while
  $\sem{\modal{a}^{0.5}1}(y)=0.25$.

\end{document}